\documentclass[12pt]{article}
\usepackage{geometry}
\usepackage[title]{appendix}
\usepackage{natbib}
\usepackage[T1]{fontenc}
\usepackage{braket}
\usepackage{setspace}
\usepackage{titlesec}
\usepackage[dvipsnames]{xcolor}
\definecolor{acmilan}{HTML}{b52e2b}
\definecolor{mediumpersianblue}{rgb}{0.0, 0.4, 0.65}
\definecolor{pblue}{rgb}{0.11, 0.22, 0.73}
\usepackage{hyperref}
\hypersetup{hypertexnames=false,colorlinks=true, citecolor=pblue, linkcolor=pblue, urlcolor=pblue}

\usepackage{amsmath,amsthm,amssymb,ifthen,mathrsfs,amsfonts}

\usepackage{empheq}
\usepackage{mathtools}
\usepackage{graphicx}
\usepackage{accents}
\usepackage{enumitem}
\usepackage{bigints}
\usepackage{tikz}

\usetikzlibrary{arrows,intersections}

\newtheorem{theorem}{Theorem}
\newtheorem{lemma}[theorem]{Lemma}
\newtheorem{proposition}[theorem]{Proposition}
\theoremstyle{definition}

\newtheorem{definition}[theorem]{Definition}

\newtheorem{remark}[theorem]{Remark}
\newtheorem{assumption}[theorem]{Assumption}

\newcommand{\BE}{\mathsf{E}}
\newcommand{\BP}{\mathsf{P}}

\newcommand{\BR}{\mathbb{R}}
\newcommand{\BN}{\mathbb{N}}

\newcommand{\mc}[1]{\mathcal{#1}}

\renewcommand{\d}{\mathrm{d}}
\newcommand{\ve}{\varepsilon}

\newcommand{\angbrac}[2]{\langle #1,#2 \rangle}

\newcommand{\mb}[1]{\mathbf{#1}}
\newcommand{\ovl}[1]{\overline{#1}}
\newcommand{\unl}[1]{\underline{#1}}

\numberwithin{equation}{section}

\makeatletter
\newcommand{\rn}[1]{\romannumeral #1}
\newcommand{\Rn}[1]{\expandafter\@slowromancap\romannumeral #1@}
\makeatother

\title{Robust Experimentation in the Continuous Time Bandit Problem}
\date{}
\author{Farzad Pourbabaee\footnote{Email: \href{mailto:farzad@berkeley.edu}{farzad@berkeley.edu}}}

\begin{document}
\maketitle
\begin{abstract}
	We study the experimentation dynamics of a decision maker (DM) in a two-armed bandit setup (\cite{bolton1999strategic}), where the agent holds \textit{ambiguous} beliefs regarding the distribution of the return process of one arm and is certain about the other one. The DM entertains \textit{Multiplier preferences} \`a la \cite{hansen2001robust}, thus we frame the decision making environment as a two-player differential game against \textit{nature} in continuous time. We characterize the DM's value function and her optimal experimentation strategy that turns out to follow a cut-off rule with respect to her belief process. The belief threshold for exploring the ambiguous arm is found in closed form and is shown to be increasing with respect to the ambiguity aversion index. We then study the effect of provision of an unambiguous information source about the ambiguous arm. Interestingly, we show that the exploration threshold rises unambiguously as a result of this new information source, thereby leading to more \textit{conservatism}. This analysis also sheds light on the efficient time to reach for an expert opinion.
\end{abstract}

{\let\thefootnote\relax\footnote{414 Evans Hall, University of California, Berkeley, CA 94720}}


{\let\thefootnote\relax\footnote{I would like to thank Robert Anderson, Philipp Strack, Gustavo Manso and Demian Pouzo for the support and guidance over the course of this paper, and I am grateful to Haluk Ergin, Chris Shannon and David Ahn for the valuable comments and suggestions. All remaining errors are mine.}}

\setcounter{footnote}{0}
\noindent {\it JEL classification:} C44; C61; C73; D81

\noindent {\it Keywords:}  Model uncertainty; Dynamic experimentation; Variational preferences; Information valuation; Ambiguous diffusion 
\section{Introduction}
\label{sec: intro}

There are natural cases where the experimentation shall be performed in ambiguous environments, where the distribution of future shocks is unknown. For example, consider a diagnostician who has two treatments for a particular set of symptoms. One is the conventional treatment that has been widely tested and has a known success rate. Alternatively, there is a second treatment that is recently discovered and is due to further study. The diagnostician shall perform a sequence of experiments on patients to figure out the success/failure rate of the new treatment. However, the adversarial effects of the mistreatment on certain types of  patients are fatal, thus the medics must consider the \textit{worst-case} scenario on the patients while evaluating the new treatment. As another case, consider the R\&D example of \cite{weitzman1979optimal}, where the research department of an organization is assigned with the task of selecting one of the two technologies producing the same commodity. The research division holds a prior on the generated saving of each technology, but the observations of each alternative during the experimentation stage is obfuscated by ambiguous sources such as the quality of researchers and managerial biases toward one choice. Therefore, the technology that is selected and sent to the development stage must be robust against these sources, because once developed it will be then used in mass production, thus even minor miscalculations in the research stage can lead to huge losses in the sales stage relative to what could have been possibly achieved.

At the core of our paper is an experimentation process between two projects framed as a two-armed bandit problem. The return rate to one arm is known to be $r$, whereas the return rate of the second arm is a binary random variable $\theta \in \{\unl{\theta},\ovl{\theta}\}$, such that $\unl{\theta} \leq r \leq \ovl{\theta}$. The decision maker (henceforth DM) holds an initial prior $p_0=\BP\left(\theta =\ovl{\theta}\right)$, that can be updated when she invests in the second project and learns its output. At the outset, she has to sequentially choose arms to learn about the unknown return rate while maximizing her net experimentation payoff. Specifically in our model, the observations of the second arm are obfuscated by Wiener process whose distribution, from the perspective of the DM, is unknown and therefore is called the \textit{ambiguous} arm. Central to the agent's decision making problem is her preference for robustness against a candidate set of future shocks' distribution which are concealing the ambiguous arm's return rate. Our investigation of the multiplicity of shocks' distribution is motivated both from the subjective and objective perspectives. Subjectively, the DM might be ambiguity averse and the multiple prior set (for the shock distribution) would be part of her axiomatic utility representation (\cite{gilboa1989maxmin}). Alternatively, the DM might be subject to an experimentation setup where the results are objectively drawn from a family of distributions, and she wants to maintain a form of robustness against this multiplicity; this is along the lines of \textit{model-uncertainty} pioneered by \cite{hansen2001robust} and \cite{hansen2006robust}.

\subsection{Summary of results}
\label{sub:summary_of_results}
We frame the decision making environment in which the DM has Multiplier preference, \`a la \cite{hansen2001robust}, as a two-player continuous time differential game against nature --- second player. The DM's goal is to find an allocation strategy between two arms that maximizes her payoff under the distribution picked by the nature. We express the (first player's) payoff function with respect to two control processes: (\rn{1}) DM's allocation choice process between the two arms, and (\rn{2}) the nature's adversarial choice of underlying distribution. The DM follows the \textit{max-min} strategy, namely at every point in time she chooses her allocation weights between two arms, and then the nature picks the shock distribution that minimizes the DM's continuation payoff. We then characterize the value function (to the DM) as a solution to a certain HJBI (Hamilton-Jacobi-Bellman-Isaac) equation.

In this game, the nature's move, i.e choice of the shock distribution, would have two important impacts (with opposite forces) on the DM. First, it affects the current flow payoff of experimentation, and secondly it distorts the DM's posterior formation and consequently her continuation strategy. In the equilibrium the DM knows the nature's best-response strategy, therefore, when she Bayes-updates her belief about $\theta$, she is no longer concerned about \textit{all} possible distributions of shocks. This gives rise to a unique law of motion for the posterior process, and reduces the HJBI equation to a second order HJB equation. 

We derive a closed-form expression for the DM's value function with respect to her posterior, i.e $v(p)$, and characterize her robust optimal experimentation strategy. It turns out in the equilibrium her strategy follows a cut-off rule with respect to her belief. Specifically, she switches to the safe arm from the ambiguous arm whenever her posterior drops below a certain threshold $\bar{p}$. We also find a closed-form equation for the cut-off value that allows us to perform a number of comparative statics. In particular, the threshold for selecting the ambiguous arm unambiguously rises as the DM's ambiguity aversion index increases.\footnote{The direction of such a response is intuitive, however, the sharp characterization of the threshold via the means of continuous time techniques provides us with the extent of this response.} Also, we establish that the marginal value of receiving \textit{good news} about $\theta$ is increasing, namely $v''(p) \geq 0$.

We then explore the effect of an additional unambiguous information source. In particular, we are interested to know what happens when for e.g the experimentation unit hires an expert to release risky but unambiguous information about $\theta$. The new value function $\tilde{v}(p)$ is obtained in closed-form, and the DM's optimal strategy again turns out to follow a cut-off rule (with a different threshold $\tilde{p}$). Interestingly, we show that under any circumstances, compared to the previous case the value of cut-off rises as a result of the extra information, i.e $\tilde{p} \geq \bar{p}$. Therefore, it is interpreted as though the DM becomes more conservative against choosing the second arm when offered with such information. Lastly, we show the surplus $\tilde{v}(p)-v(p)$ generated by the expert attains its maximum at the range of beliefs where the experimentation unit would otherwise select the ambiguous arm but do not have strong enough feeling and evidence in favor of this decision. Therefore, our model sheds light on the time that is best to reach an expert opinion.

\subsection{Related literature and organization of the paper} 
\label{sub:related_literature}
The literature on robust bandit problem is limited, but recently there have been some attempts to bring several aspects of robustness into play. Specifically in the works done by \cite{caro2013robust} and \cite{kim2015robust} the discrete-time multi-armed bandit problem is studied while the state transition probabilities are drawn from an \textit{ambiguous} set of conditional distributions. In \cite{caro2013robust} the set of multiple transition probabilities at every period is constrained through a relative entropy condition, whereas \cite{kim2015robust} chooses to impose an entropic penalty cost directly in the objective function of the DM rather than hard thresholding it as a constraint. In a different work \cite{li2019k} studies the multi-armed bandit in which the DM entertains max-min utility and follows a prior-by-prior Bayes updating from her initial \textit{rectangular} multiple prior set, where each candidate distribution in this set is identified by the i.i.d shocks it generates in the future. Our work is different from these treatments in the following aspects: (\rn{1}) contrary to the first two works the Brownian diffusion treatment of the Markov transitions allows for a richer set of perturbations around benchmark model which extends the scope of robustness that the DM demands; (\rn{2}) the continuous time framework lets us to obtain sharp and closed form results on the value function and the optimal experimentation policy that in turn renders the comparative static with regard to parameters of the model and importantly the ambiguity aversion index; (\rn{3}) we are explicit about the state variable in our setup, and specifically we characterize it as the DM's posterior process regarding the second arm's return rate; (\rn{4}) our setup is flexible enough that can address distinct informational environments such as the effect of the provision of an expert opinion.

In the economic literature, after the seminal work of \cite{gittins1979bandit}, the continuous time problem of optimal experimentation in a noisy environment, where the payoff to the unexplored arm\footnote{Often the second arm is referred as the \textit{unexplored} one.} is subject to a Brownian motion is studied in \cite{bolton1999strategic} and \cite{keller1999optimal}. Aside from these works, there is a growing literature on experimentation in a multiple agent environment where the free-riding issues arise.\footnote{A nonexhaustive list includes \cite{keller2005strategic}, \cite{heidhues2015strategic} and \cite{bonatti2017learning}.}

Our treatment of robust preferences in continuous time relies heavily on the fundamental works by \cite{hansen2001robust}, \cite{hansen2006robust} and \cite{hansen2011robustness}.\footnote{In a closely related discrete-time framework \cite{epstein2003recursive} and \cite{maccheroni2006dynamic} present recursive utility representation aimed to capture the preference for robustness.} Our paper is also related to the literature studying the effects of robustness and ambiguity in different decision making frameworks such as \cite{riedel2009optimal}, \cite{cheng2013optimal}, \cite{miao2016robust}, \cite{wu2018ambiguity} and \cite{luo2017robustly}. Also, it is related to the relatively understudied topic of learning under ambiguity.\footnote{For example see \cite{marinacci2002learning}, \cite{epstein2007learning} and \cite{epstein2019optimal}.} Finally in a set of experimental works with adopting different notions of ambiguity aversion, it has been tested that the ambiguous arm of the experiment has a lower Gittins index that prompts the DM to undervalue the information from exploration. To name a few we can point to \cite{anderson2012ambiguity} and \cite{meyer1995sequential} in the context of airline choice and \cite{viefers2012should} in the investment choice.

The remainder of the paper is organized as follows. To build intuition, in section \ref{sec: two_period_example} we present some of the forces behind the model in a two-period example. Next, in section \ref{sec: model} the full features of experimentation setup and payoff function are explained in a continuous time framework. In section \ref{sec:dynamic_programming_analysis}, we apply the dynamic programming analysis and present variational characterizations of the value function. Section \ref{sec: valuefunction} offers the closed-form expression for value function, properties of the optimal experimentation strategy, and some comparative static results. In section \ref{sec:value_of_unambiguous_information}, we extend our setup to capture the effect of an additional unambiguous information source. The concluding remarks are presented in section \ref{sec:concluding remarks} and finally the proofs of all results are expressed in the appendix \ref{sec: proofs}.
\section{Two-period example}
\label{sec: two_period_example}
Our goal in this example is to highlight the main trade-offs that the DM and her opponent \textit{nature} face in their dynamic interaction. Let $t \in\{1,2\}$ and at each period the DM allocates her resources between two available choices, namely the safe and the ambiguous project. The time $t$ incremental returns to each arm when she allocates $\mu_t \in [0,1]$ of her resources to the safe (first) arm and $1-\mu_t$ to the ambiguous (second) arm are
\begin{equation}
	\begin{split}
		\Delta y_{1,t}&=(1-\mu_t) r\\
		\Delta y_{2,t}&=\mu_t\theta+\sqrt{\mu_t}\ve_t.
	\end{split}
\end{equation}
In that $r=1$ is the return rate of the safe project, and $\theta \in \{0,2\}$ is the unknown return to the second arm. The DM's prior on this set at period one is given by $p_1 = \BP\left(\theta=2\right)$, which is not subject to any ambiguity. However, at each period the return to the second arm is obfuscated by an independent\footnote{For simplicity assume $\ve_1$, $\ve_2$ and the period one belief on $\theta$ are independent from each other.} Gaussian shock that could possibly be drawn from two distributions, namely for each $t$ the law of $\ve_t$ belongs to the set $\left\{\mc{N}(-0.5,1),\mc{N}(0.5,1)\right\}$.\footnote{This set clearly doesn't satisfy the rectangularity condition nor the convexity property of \cite{gilboa1989maxmin}, however it serves only for expositional purposes.} We take no stance on whether this multiple prior set is the subjective belief of the DM or literally the objective moves that nature takes against the DM. Our solution concept for both cases is the the so-called \textit{max-min}. However, the first situation reflects a decision theoretic choice of an ambiguity averse agent with a subjective multiple prior set, whereas the second interpretation is more in line with the notion of robust decision making. 

The timing of this example is as follows. At the beginning of period one DM chooses $\mu_1$. Then, nature \textit{responds} by picking $h_1 \in \{-0.5,0.5\}$ as the mean of $\ve_1$. The returns to both arms, i.e $\{\Delta y_{1,1},\Delta y_{2,1}\}$ are realized. DM forms the family of beliefs $\{p_2^{h_1}: h_1\}$ at the beginning of period two, and takes the \textit{appropriate} action $\mu_2$. The nature chooses $h_2$ as the mean of second period's shock. Subsequently the game ends and second period's returns are realized.

What happens at the sub-game perfect equilibrium of this game? For this we need to look at the sub-game starting at $t=2$. Regardless of DM's action $\mu_2$, the nature always picks $h_2 = -0.5$, because the game ends at this period and $h_2=-0.5$ is the worst case distribution from the DM's perspective. Because of this triviality of the nature's choice at period two, we drop the index one from $h_1$ and henceforth denote it by $h$, which is the only non-trivial choice of the nature in this example. The DM's posterior beliefs after the realizations of first period returns are
\begin{equation}
\label{eq: p2update}
p_2^h = \left(1+\frac{1-p_1}{p_1}\exp\left\{2\left(\sqrt{\mu_1}h+\mu_1-\Delta y_{2,1}\right)\right\}\right)^{-1}1_{\{\mu_1>0\}}+p_11_{\{\mu_1=0\}}, ~~ h \in \{-0.5,0.5\}.
\end{equation}
It is important to note that the posterior probability is no longer unique, and DM faces a set of posteriors for each choice of nature in period one. Even though that we face a two-player game where the nature's actions are not observable to the DM, but at the equilibrium DM knows the \textit{minimizing} choice of the nature, thereby her family of posteriors effectively reduces to a single posterior induced by the worst case action of the nature say $h^*$. This point becomes more clear as we proceed through the equilibrium analysis. For every member $p_2$ of the posterior set, the DM's optimal action at period two (anticipating that nature will choose $h_2=-0.5$) is $\mu_2(p_2) = 1_{\{2p_2-0.5 > 1\}}$, that leads to the expected payoff of $v_2(p_2) = \max\{1,2p_2-0.5\}$. Note that this expectation is with respect to the equilibrium distribution choice of the nature that is $h_2=-0.5$. Assume the experimenter's intertemporal discount rate is $\delta \in (0,1]$. Further, let $\BP^h$ denote the probability measure induced by the independent product of $\ve_1 \sim \mc{N}(h,1)$ and $\theta \sim p_1$. Therefore, the DM's value function as of beginning of period one is
\begin{equation}
\label{eq: 2periodvalue}
v_1(p_1) = \max_{\mu_1 \in [0,1]} \min_{h \in \{-0.5,0.5\}}\left\{\left[(1-\mu_1)+2\mu_1p_1+\sqrt{\mu_1}h)\right]+ \delta\BE^{h} \left[v_2(p_2^h)\right]\right\}.
\end{equation}
Below we point out to some of the underlying equilibrium forces that will show up in this two period example.
\begin{enumerate}[leftmargin=*,label=(\roman*)]
    \item The nature's first period action, or alternatively, the most pessimistic perception of the DM in regard to shock distribution $\ve_1^h$, plays two roles. \textbf{Current payoff channel}, in that the nature's choice of $h$ affects the current payoff of the DM by changing the mean return of the ambiguous arm, i.e $\left[(1-\mu_1)+2\mu_1p_1+\sqrt{\mu_1}h)\right]$. In particular, this is a \textit{positive} force, as higher $h$'s correspond to higher mean flow payoff.
    \textbf{Informational channel}, where the shock distribution $\ve_1^h$ affects the next period belief of the DM, hence changes her course of action and thereby the continuation payoff. This has a \textit{negative} effect, because as $h$ increases, the distribution of $\Delta y_{2,1}$ shifts to the right in the FOSD sense and for a fixed $\Delta y_{2,1}$ lowers the likelihood ratio in \eqref{eq: p2update} that in turns depresses the continuation payoff $\BE^{h} \left[v_2(p_2^h)\right]$. At the equilibrium, nature counteracts these forces and picks the one that its negative effect outweighs the positive one, and thus reduces the DM's payoff more. However, it can not completely balance out the marginal impact of these forces, mainly because we assumed the multiple prior set consists of only two distributions. When the complete mode is laid out in section \ref{sec: model}, we allow for quite general multiple prior set, thus nature can precisely cancel out the marginal effects, thereby lowering the DM's payoff as much as possible. 

    \item From the point of view of the DM, there is an option value of experimentation. Specifically, in the first period she selects the ambiguous arm (even partially $0<\mu<1$) only to observe the payout of second arm, and then may decide to abandon the ambiguous project depending on the outcome of the first period. In this example, the DM switches back to the safe arm in the second period if her posterior in the equilibrium, i.e $p_2^{h^*}$, drops below a certain threshold, which in this case is $0.75$.

    \item The DM's value function is unambiguously increasing in her initial belief $p_1$ (as can be confirmed from \eqref{eq: 2periodvalue}), but the marginal value of good news need not be increasing (meaning $v''$ is not always positive). This is mainly due to the finite-horizon setup of the two-period model, which is relaxed in later sections.
    
    \item The value function in \eqref{eq: 2periodvalue} refers to the max-min value of the game, which is associated to the strategic order of actions in which the DM takes her action first and then the nature responds in every period. This is the same approach that we pursue when we present the complete model. However, one might wonder when does this max-min value coincide with the min-max one? Or in the other words, when does the strategic order of players' actions become irrelevant? In this example the max-min value is strictly less than min-max. Although not related to the study of this paper, but we confirm that with compact and convex action spaces of both players, the von-Neumann minimax theorem could be applied and therefore one can conceive the unique value of the zero-sum game between DM and the nature.
\end{enumerate}
We do not intend to delve deeper into this example and express more specific results and comparative statics, mainly because such analysis will be carried out for the complete model later in the paper.

\section{Experimentation model}
\label{sec: model}
Time horizon is infinite and $t \in \BR_+$. There are two projects available to experiment by the DM. Her choice at time $t$ is thus to allocate her resources between two alternatives, namely $\mu_t$ to the ambiguous arm and $1-\mu_t$ to the safe arm. The return process of the projects are\footnote{The goal of this section is to study the interplay between ambiguity regarding the new arm and optimal experimentation, thus for simplicity we assume that the conventional arm has a sure return rate of $r$ and is not subject to any source of randomness. Therefore, it is only the second arm that carries the Brownian motion term.}
\begin{equation}
\label{eq: payoffprocesses}
	\begin{split}
		\d y_{1,t} &= (1-\mu_t)r \d t\\
		\d y_{2,t} &= \mu_t \theta \d t +\sigma\sqrt{\mu_t} \d B_t.
	\end{split}
\end{equation}
Here $B$ is a Brownian motion relative to some underlying stochastic basis\footnote{The description of the underlying stochastic basis and the joint structure of processes are explained in the subsection devoted to the \textit{weak formulation}.}, that represents the shock process, and $\theta$ is unknown to the DM but belongs to the binary set $\{\bar{\theta},\underline{\theta}\}$, where $\underline{\theta} \leq r \leq  \bar{\theta}$. The DM has an initial belief $p_0=\BP\left(\theta=\overline{\theta}\right)$ about $\theta$ which is independent from $B$. The form of return processes in \eqref{eq: payoffprocesses} follows \cite{bolton1999strategic}, but we let the DM to associate multiple distributions to the shock process. Specifically, the DM holds a single belief over $\theta$ --- so that this represents the uncertainty due to \textit{risk} --- but has multiple beliefs regarding the shock distribution $B$ --- so this represents the uncertainty due to \textit{ambiguity}.\footnote{This type of uncertainty is sometimes referred to as \textit{model uncertainty} in the literature.}
\subsection{A framework for modelling ambiguity}
Our take of ambiguity or model uncertainty is similar to \cite{hansen2006robust} and \cite{hansen2011robustness}. In particular, we assume there is a family of pairs $\left\{(\BP^h,B^h): h\in \mc{H}\right\}$ such that for each $h \in \mc{H}$, $B^h$ is a Brownian motion under $\BP^h$, and DM views this as her multiple prior set. We think of $\mc{H}$ -- which thus far has not been defined -- as the nature's action space, and each $h\in \mc{H}$ is deemed as a possible nature's move. We assume there exists a \textit{benchmark} probability specification $\BP$ that is \textit{equivalent} (mutually absolutely continuous with respect) to each member of $\mc{P}:=\{\BP^h: h\in \mc{H}\}$. The benchmark measure $\BP$ and the set $\mc{P}$ are interpreted differently based on the context. For example, DM might believe that $\BP$ is the underlying probability measure, but considers $\mc{P}$ as the approximations of the true distribution because she has preference for robustness. Alternatively, $\mc{P}$ could be conceived as the multiple prior set for the ambiguity averse DM in the axiomatic treatment of \cite{gilboa1989maxmin}.

DM has \textit{Multiplier preference} and maximizes the following payoff over an \textit{admissible} set of experimentation strategies $\mc{U}$ --- with some technical considerations that are elaborated later in the paper:
\begin{equation}
\label{eq: multu}
	\inf_{h \in \mc{H}} \left\{\BE^{\BP^h}\left[\delta \int_0^\infty e^{-\delta t} \d \left(y_{1,t} + y_{2,t}\right)\right]+\alpha  H\left(\BP^h; \BP\right)\right\}
\end{equation}
Here $\delta$ is the time discount rate. The first term in the DM's utility is simply the expected discounted payoff from both projects taken with respect to the measure $\BP^h$, and the second term penalizes the belief misspecification using the relative discounted entropy to measure the discrepancy between $\BP$ and $\BP^h$. Parameter $\alpha$ captures the extent of this penalization, where its larger values associate to smaller penalty. We shall also interpret $\alpha$ as the inverse of ambiguity aversion and relate \eqref{eq: multu} to the \textit{dynamic variational utility representation} of \cite{maccheroni2006ambiguity} and \cite{maccheroni2006dynamic}. A large $\alpha$ means that the DM does not suffer a lot from ambiguity aversion. In contrast as $\alpha \to 0$, the DM experiences larger utility loss due to severe penalization.

In the next subsection we use the \textit{weak-formulation} approach from the theory of stochastic processes to elaborate and simplify DM's utility function \eqref{eq: multu}.
\subsection{Weak formulation}
\label{sub:weak_formulation}
In this part we present a sound foundation for the joint structure of all the stochastic processes in the model\footnote{The materials in this subsection might look somewhat technical and unnecessary to some readers, but are essential for rigorous development of the model.}. Let $\left(\Omega, \mc{F}=\mc{F}_\infty, \mb{F}=\{\mc{F}_t\}_{t\in \BR_+},\BP\right)$ be the stochastic basis, where the filtration satisfies the \textit{usual conditions}.\footnote{It is right-continuous and $\BP$-complete.} The average rate of return to the ambiguous project $\theta$ is a binary $\mc{F}_0$-measurable random variable.

\begin{definition}[Strategy spaces]
\label{def: strategy_space}
The DM's strategy space $\mc{U}$ --- with a representative point $\mu\in \mc{U}$ --- is the set of all $\mb{F}$-progressive processes\footnote{We refer to \cite{karatzas2012brownian} for the definition of progressive processes.} taking value in $[0,1]$. The nature's strategy space $\mc{H}$ --- with a representative point $h\in \mc{H}$ --- is the space of all \textit{bounded} $\mb{F}$-progressive processes.
\end{definition}

\begin{definition}[Integral forms]
For any pair of processes $\{f,g\}$ where $f$ is $g$-integrable\footnote{The notion of integral depends on the context that could either be the path-wise Stieltjes integral or stochastic It\^o integral.} we use the alternative notation for integration: $(f \cdot g)_t: = \int_0^t f_s\d g_s$. Further, the symbol $\imath$ refers to identity mapping $t\mapsto t$ on $\BR_+$. Then the differential return expressions in \eqref{eq: payoffprocesses} can be represented in the integral form $y_1 = (1-\mu)r \cdot \imath$ and $y_{2} = \mu\theta \cdot \imath+\sigma \sqrt{\mu} \cdot B$.
\end{definition}
To model the ambiguity we appeal to the weak formulation. In particular, we think of ambiguity as the source that changes the distribution of return process $\{y_1,y_2\}$, but not its sample paths. For this on every finite interval $[0,T]$ we define the probability measure $\BP^h_T$ with the following Radon-Nikodym derivative process:
\begin{equation}
\label{eq: RNfinite}
	\left.\frac{\d \BP^h_T}{\d \BP}\right|_{\mc{F}_t}:=L_{t,T}^h=\exp\left\{\left(h \cdot B\right)_t -\frac{1}{2}\left(h^2 \cdot \imath\right)_t\right\}, \  \  \    \forall t \leq T
\end{equation}
This relation explains how nature with its choice of $h\in \mc{H}$ could induce a new probability measure. The Girsanov's theorem implies that $\BP^h_T$ is mutually absolutely continuous with respect to $\BP$ --- that is often called \textit{equivalent} measure and denoted by $\BP^h_{T} \sim \BP$ on $\mc{F}_T$. It also implies that the \textit{mean-shifted} process $B^h:=B - (h\cdot \imath)$ is a $\mb{F}$-Brownian motion under $\BP^h_T$ over the interval $[0,T]$. The main catch here is that we can only characterize the perturbations of benchmark probability model $\BP$ over finite intervals, that is for example we know how $\BP^h$ looks like on $\mc{F}_T$ for any finite $T$. However, what is needed for the utility representation in \eqref{eq: multu} is a specification of $\BP^h$ on the terminal $\sigma$-field $\mc{F}_\infty$. For this we need to use a limiting argument to consistently send $T \to \infty$ and obtain $(\BP^h,L^h,B^h)$ as an appropriate limit of $(\BP^h_T,L^h_T,B^h_T)$. Our proposal for this is as follows. For any process $h \in \mc{H}$ and an increasing sequence of finite times $\{T_n\}_{n \in \BN}$, we repeatedly apply the Girsanov's theorem to obtain a family of consistent probability measures $\left\{\BP^h_{T_n},\mc{F}_{T_n}: n \in \BN\right\}$, where $\BP^h_{T_n}\sim \BP$ on $\mc{F}_{T_n}$ for every $n \in \BN$. In a similar vein we obtain the likelihood ratio process $\left\{L^h_{t,T_n}: t\leq T_n\right\}$ and the Brownian motion $\left\{B^h_{t,T_n}: t \leq T_n\right\}$ for every $n \in \BN$. Next, we explain how to naturally define the limit of each three components.
\begin{enumerate}[leftmargin=*,label=(\roman*)]
    \item \textbf{Likelihood process limit:} Expression \eqref{eq: RNfinite} implies that the sequence of likelihood processes are path-wise consistent with each other, i.e $L_{t,T_m}^h = L_{t,T_n}^h$ for every $t \leq T_m \leq T_n$. Therefore, one can define the process $L^h$ on $[0,\infty)$ in a meaningful sense, such that its restriction to any finite interval coincides with the sequence of likelihood processes. This concludes the construction of the limit likelihood process. Importantly, this construction suggests that $L^h$ must be a martingale process with respect to $\BP$ on $\BR_+$. To see this, note that a bounded $h$ causes the \textit{Novikov's} condition to hold, thereby $L^h_{T_n}$ would be an uniformly integrable martingale --- on $[0,T_n]$ --- with respect to $\BP$ for every $n \in \BN$. Because of the path-wise equivalence, this would immediately establish the martingale property of $L^h$ on $\BR_+$.
    
    \item \label{item: prob_cons}\textbf{Probability measure limit:} First, recall that for every $n\in \BN$, $\BP^h_{T_n}$ is a probability measure on $\mc{F}_{T_n}$. Then, the path-wise consistency resulted from \eqref{eq: RNfinite} implies that these measures indeed match each other, namely $\BP^h_{T_n}(A) = \BP^h_{T_m}(A)$ for every $A \in \mc{F}_{T_m}$ where $m\leq n$. Thus, we can apply theorem 4.2 in \cite{parthasarathy2005probability} that guarantees the existence of a \textit{closing} probability measure $\BP^h$ on $\mc{F}_\infty$ such that its restrictions to finite intervals coincide with the above sequence of probability measures, yet it need not be equivalent to $\BP$ on $\mc{F}_\infty$. That is restricted to every finite $T$, $\BP^h \sim \BP$ on $\mc{F}_T$, but this may not be true on $\mc{F}_\infty$.
    
    \item \textbf{Brownian motion limit:} Applying Girsanov's theorem lets us to deduce that $B^h_{T_n}:=\{B^h_{t,T_n}:t\leq T_n\}$ is a Brownian motion under $\BP^h_{T_n}$ on $[0,T_n]$ for every $n \in \BN$. Since $\BP^h \equiv \BP^h_{T_n}$ on $[0,T_n]$, then it turns out that $B^h_{T_n}$ is also a Brownian motion under $\BP^h$. Also note that the path-wise consistency holds for the sequence of Brownian motions, namely $B^h_{t,T_n}= B^h_{t,T_m}$ for all $t\leq T_m \leq T_n$. Therefore, in the same manner that we defined $L^h$ from $\{L^h_{T_n}:n \in \BN\}$, we can define $B^h$ as the process on $\BR_+$ such that its restrictions to any finite interval satisfy the properties of $\BP^h$ Brownian motions. 
\end{enumerate}

The illustrated construction of $\left(\BP^h,B^h\right)$ allows us to express the return process of the ambiguous project in term of $h$-Brownian motion:
\begin{equation}
\label{eq: y2muh}
	\d y_{2,t} = \left[\mu_t\theta + \sigma \sqrt{\mu_t}h_t\right]\d t + \sigma \sqrt{\mu_t}\d B_t^{h}
\end{equation}
The merit of weak formulation now becomes clear, where for every $\mu \in \mc{U}$ the return processes $\{y_1,y_2\}$ are essentially fixed, but the probability distribution that assigns weights to the subsets of sample paths is controlled by the choice of $h \in \mc{H}$. So in a sense the nature's move is to select the return's distribution not its sample paths.

Now that we know what is meant by $\BP^h$ on $\mc{F}_\infty$ we can analyze both terms of \eqref{eq: multu} which are expectations under $\BP^h$, and this will be the goal of next subsection.

\subsection{Unravelling the payoff function}
\label{sub: payoff}
We begin the simplification of \eqref{eq: multu} by elaborating the second term, that is the entropy cost of ambiguity aversion. Recall that $\BP$ and $\BP^h$ need not necessarily be equivalent measures on $\mc{F}_\infty$, yet their restrictions $\{\BP_t,\BP^h_t\}$ are indeed equivalent probability measures on $\mc{F}_t$. Having that said, the relative discounted entropy is defined as
\begin{equation}
\label{eq: disrelentr}
 H\left(\BP^{h} ; \BP\right) := \lim_{T\to \infty} \delta \int_0^T e^{-\delta t} H\left( \BP^{h}_t; \BP_t\right)\d t,
\end{equation}
where $H\left(\BP_t^h;\BP_t\right):=\BE^{h}\left[\log L_t^{h}\right]$. Expression \eqref{eq: disrelentr}, which is proposed in \cite{hansen2006robust}, presents a proxy for the discrepancy between two measures that are not necessarily equivalent on the terminal $\sigma$-field, and hence their relative entropy $H\left(\BP^h_\infty;\BP_\infty\right)$ could be infinite, but on each finite interval say $[0,t]$ they are equivalent $\BP^h_t \sim \BP_t$ and have finite relative entropy. Therefore, one shall hope that relation \eqref{eq: disrelentr} is well-defined.

\begin{lemma}
\label{lem: wellentr}
     The discounted relative entropy in \eqref{eq: disrelentr} is \textit{well-defined}, namely for every $h\in \mc{H}$ it is finite and satisfies
 	\begin{equation}
 		H\left(\BP^{h}; \BP\right)= \frac{1}{2} \BE^{h}\left[\int_0^\infty e^{-\delta t} h^2_t \d t\right]<\infty.\
 	\end{equation}
\end{lemma}
Roughly speaking, for the first component of the payoff function we need to take the expectation of $\d y_2$ under the measure $\BP^h$. This is in our reach because we stated the dynamics of $y_2$ in terms of $B^h$ in \eqref{eq: y2muh}.  However, the drift term in $\d y_2$ contains the random variable $\theta$, that needs to be learned and projected onto the DM's information set. For this we present an optimal filtering result under each measure $\BP^h$. 
\begin{remark}
The DM's initial prior $p_0=\BP\left(\theta =\ovl{\theta}\right)$ is unaffected under different probability distributions $\mc{P}=\left\{\BP^h: h \in \mc{H}\right\}$. This is because the benchmark measure $\BP$ and all its variations $\mc{P}$ agree on $\mc{F}_0$, resulted from $L_0^h=1$ for every $h \in \mc{H}$.
\end{remark}
In light of this remark, we want to continuously estimate and update the DM's posterior on $\theta$ based on her available information at every point in time. Her information set at time $t$ contains the path of output from each project $\left\{\left(y_{1,s},y_{2,s}\right): s\leq t\right\}$, the history of her allocation process $\left\{\mu_s: s\leq t\right\}$ and importantly the nature's moves up until time $t$, i.e $\left\{h_s: s\leq t\right\}$. Note that at each time $t$, the DM's ambiguity is with regard to the future path of $h$, and she has no uncertainty about the history of nature's moves in the past. Some might not be willing to make this assumption about the ex-post observability of nature's moves to the DM. However, this is not an important assumption for two reasons. First, on the equilibrium path the DM knows the history of nature's past moves. Secondly, in theory we can find the filtering equation under every possible history of nature's actions and then let the DM to pessimistically choose from this family of posteriors. In summary, the filtering problem that the DM faces at time $t$ is to update her posterior based on the available information set $\mc{F}^{y_1,y_2, \mu,h}_t$. Of secondary importance is to note that $y_1$ conveys no information about $\theta$, thus can be dropped out of the information set.
\begin{definition}
For every $t \in (0,\infty)$, define $p_t^h:=\BP^h\left(\left.\theta = \bar{\theta} \right|\mc{F}_t^{y_2,\mu,h}\right)$ as the posterior probability and $m(p_t^h) = p_t^h \bar{\theta}+(1-p_t^h)\underline{\theta}$ as the conditional mean. At $t=0$, let $p_t^h=p_0$ and $m(p_0^h)=m(p_0)$.
\end{definition}
\begin{lemma}[\cite{liptser2013statistics} theorem 8.1]
\label{lem: optfilt}
The conditional probability of the event $\{\theta=\ovl{\theta}\}$ given the filtration $\mb{F}^{y_2,\mu,h}$ evolves according to the following stochastic differential equation:
\begin{equation}
\label{eq: beliefmg}
    \d p_t^{h}=\frac{(\bar{\theta}-\underline{\theta})\sqrt{\mu_t}}{\sigma}p_t^{h}\left(1-p_t^{h}\right)\d \bar{B}_t^{h}
\end{equation}
Here $\left\{\bar{B}^h_t, \mc{F}^{y_2,\mu,h}_t: t\in \BR_+\right\}$ is called the innovation process which is a Brownian motion under $\BP^h$, and is characterized by $\d \bar{B}^h_t = \sigma^{-1}\sqrt{\mu_t}\left[\theta-m(p^h_t)\right]\d t+ \d B^h_t$. As a result of this, the law of motion for $y_2$ would be
\begin{equation}
\label{eq: filteredy2}
	\d y_{2,t}=\left[\mu_t m(p_t^h)+\sqrt{\mu_t}h_t\right]\d t+\sigma \sqrt{\mu_t}\d \bar{B}^h_t.
\end{equation}
\end{lemma}
\textit{Sketch of the proof.} First note that from the filtering point of view the process $y_2$ contains the same information as $\tilde{y}_2:=(\sqrt{\mu}\theta+\sigma h) \cdot \imath +\sigma \cdot B^h$. Therefore, on the region $\mu>0$, we have $\BE^h\left[\left. \theta \right|\mc{F}_t^{\tilde{y}_2,\mu,h}\right]=\BE^h\left[\left. \theta \right|\mc{F}_t^{y_2,\mu,h}\right]$ for every $h \in \mc{H}$ and $t\in \BR_+$. Next, applying theorem 8.1 of \cite{liptser2013statistics} and taking $\tilde{y}_2$ as the \textit{observable} process and $\theta$ as the subject of filtering imply that:
\begin{equation}
\label{eq: fundametal_filter}
    \begin{split}
        \BE^h\left[\left. \theta \right|\mc{F}_t^{\tilde{y}_2,\mu,h}\right] &= \BE^h\left[\left. \theta \right|\mc{F}_0^{\tilde{y}_2,\mu,h}\right]\\
        &\hspace{-40pt}+\sigma^{-1}\int_0^t \left(\BE^h\left[\left. \theta \left(\sqrt{\mu}_s\theta+h_s\right) \right|\mc{F}_s^{\tilde{y}_2,\mu,h}\right] -\BE^h\left[\left. \theta \right|\mc{F}_s^{\tilde{y}_2,\mu,h}\right] \BE^h\left[\left.\sqrt{\mu}_s\theta+h_s  \right|\mc{F}_s^{\tilde{y}_2,\mu,h}\right]\right)\d \bar{B}^h_s\\
        &=\BE^h\left[\left. \theta \right|\mc{F}_0^{\tilde{y}_2,\mu,h}\right]+\sigma^{-1}\int_0^t \sqrt{\mu_s}\left(\BE^h\left[\left. \theta^2 \right|\mc{F}_s^{\tilde{y}_2,\mu,h}\right]-\BE^h\left[\left. \theta \right|\mc{F}_s^{\tilde{y}_2,\mu,h}\right]^2\right)\d \bar{B}^h_s
    \end{split}
\end{equation}
This expression underlies the filtering equation for the posterior process $p^h$, as it readily amounts to
\begin{equation}
    \begin{split}
        p^h_t=p_0+\sigma^{-1}(\bar{\theta}-\underline{\theta})\int_0^t \sqrt{\mu_s} p_s^h(1-p_s^h)\d \bar{B}_s^h,
    \end{split}
\end{equation}
and thus verifies equation \eqref{eq: beliefmg}. It is worth mentioning here that since there is no ambiguity about $\theta$ at time $0$ w.r.t the distribution of $\theta$, the first term in the \textit{rhs} of \eqref{eq: fundametal_filter} is independent of $h$. \qed

At this stage we have developed all the required tools to present the utility function in \eqref{eq: multu} in terms of initial belief and the players' actions. For this we define the infinite horizon payoff as the limit of finite horizon counterparts. The reason is that the constructed process $B^{h}$ is only Brownian motion over finite intervals, and we can not extend it to entire $\BR_+$, unless we impose further restrictions on $\mc{H}$ and $\mc{U}$ to obtain the uniform integrability of likelihood processes, which we refrain to do. Therefore, inspired by \eqref{eq: multu} we define the utility of DM from taking action $\mu$ while nature chooses $h$ by
\begin{equation}
\label{eq: payoffdef}
V(p;\mu,h):=\lim_{T \to \infty}\BE^h\left[\delta \int_0^T e^{-\delta t} \left(\d y_{1,t}+\d y_{2,t}+\alpha H\left(P_t^h;P_t\right)\d t\right)\right].
\end{equation}
\begin{proposition}
\label{prop: payoffrep}
    For every choice of $\mu \in \mc{U}$ and $h\in \mc{H}$, the net discounted average payoff defined in \eqref{eq: payoffdef} can be expressed as:
    \begin{equation}
    	\label{eq: payoffrep}
		V(p;\mu,h)= \BE^{h}\left[\delta \int_0^\infty e^{-\delta t}\left((1-\mu_t)r+ \mu_t m(p_t^h) +\sigma \sqrt{\mu_t}h_t+\frac{\alpha}{2\delta}h^2_t\right)\d t\right]
    \end{equation}
\end{proposition}
This proposition serves us well, because the integrand is now $\mb{F}^{y_2,\mu,h}$-progressively measurable, that in turn allows us to perform a dynamic programming scheme to express the \textit{value function} in terms of the current belief, and this will be the goal of next section.
\section{Dynamic programming analysis}
\label{sec:dynamic_programming_analysis}
Our analysis so far offers expression \eqref{eq: payoffrep} as the DM's payoff in the two-player differential game against the nature. For any point of time, say $t\in \BR_+$, define the expected continuation value conditioned on $\mc{G}_t:=\mc{F}_t^{y_2,\mu,h}$ as
\begin{equation}
\label{eq: Jfunc}
	J(p,t;\mu,h):=\BE^h\left[\left. \delta \int_t^\infty e^{-\delta s}\left((1-\mu_s)r+ \mu_s m(p_s^h) +\sigma \sqrt{\mu_s}h_s+\frac{\alpha}{2\delta}h^2_s\right)\d s\right| \mc{G}_t  \right].
\end{equation}
In that $p$ is the time $t$ value of the state process $p_t^h$. For every $h \in \mc{H}$ the process $\bar{B}^h$ as well as $p^h$ are time \textit{homogeneous} Markov diffusions. Furthermore, the players' action spaces at the time $t$ sub-game --- $\mc{U}_t$ and $\mc{H}_t$ resp. for the DM and the nature --- are essentially isomorphic to $\mc{U}$ and $\mc{H}$. These two premises imply that the max-min value of the game for the DM, i.e $\sup_{\mu \in \mc{U}_t}\inf_{h \in \mc{H}_t} J(p,t;\mu,h)$,
is time homogeneous. Specifically, there exists a value function $v(p)$
such that
\begin{equation}
\label{eq: supinfdef}
	\sup_{\mu \in \mc{U}_t}\inf_{h \in \mc{H}_t}J(p,t;\mu,h)=e^{-\delta t}v(p)	
\end{equation}
Our goal in the next theorem is to present a \textit{verification} result for the value function. For this we need to appeal to the theory of viscosity solution \cite{crandall1984some} that provides the appropriate setting for Bellman equations. The reason for this is that as it turns out the value function $v(p)$ is not twice continuously differentiable everywhere, therefore classical verification techniques relying on Ito's lemma would not apply. We offer some preliminary definitions that are linked to the work of \cite{zhou1997stochastic}\footnote{There were some technical gaps in the proof of the verification theorem in this paper, that are addressed and corrected in the follow up papers \cite{gozzi2005corrected} and \cite{gozzi2010erratum}; thanks to the anonymous referee for bringing this up to the author's attention.}, thereby setting the groundwork for the viscosity solution concept.
\begin{definition}
Let $w \in C([0,1])$. The \textit{superdifferential} of $w$ at $x_0\in [0,1)$ is denoted by $D_+w(x_0)$:
\begin{equation}
    D_+w(x_0)=\left\{(\xi_1,\xi_2)\in \BR^2: \limsup_{x\to x_0}\frac{w(x)-w(x_0)-(x-x_0)\xi_1-\frac{1}{2}(x-x_0)^2\xi_2}{(x-x_0)^2}\leq 0\right\}
\end{equation}
A generic member of this set is referred by $\left(\partial_+w(x_0),\partial_+^2w(x_0)\right)$.
And the \textit{subdifferential}, denoted by $D_-w(x_0)$ is defined as
\begin{equation}
    D_-w(x_0)=\left\{(\xi_1,\xi_2)\in \BR^2: \liminf_{x\to x_0}\frac{w(x)-w(x_0)-(x-x_0)\xi_1-\frac{1}{2}(x-x_0)^2\xi_2}{(x-x_0)^2}\geq 0\right\}.
\end{equation}
A generic member of this set is referred by $\left(\partial_-w(x_0),\partial_-^2w(x_0)\right)$.
\end{definition}
Notice that a continuous function may not be once or twice continuously differentiable but it always has non-empty super(sub)-differential sets on a dense subset of $[0,1]$ \cite{lions1983optimal}.

In the \textit{verification} theorem that follows we show that the value function $v(\cdot)$ in \eqref{eq: supinfdef} is the viscosity solution to a certain HJBI equation with the following form
\begin{equation}
\label{eq: HJBI_general}
	w(p) = \sup_{\mu \in [0,1]} \inf_{h\in \BR} \left\{g(p,\mu,h) +\mathcal{K}(p,w'(p),w''(p),\mu,h)\right\}\footnote{Notice that $w'$ and $w''$ should not be confused with the first and second derivatives as they may not exist for a continuous function. This form is just a \textit{representation} of the HJBI equation that has a viscosity solution in the sense of definition \ref{def: viscosity_sol}, and may not hold a \textit{smooth} classical solution.},
\end{equation}
where the specific form of the coefficients $g$ and $\mathcal{K}$ will be given in the theorem's statement. As a last step before presenting the therorem, we express what is meant by being a viscosity solution to a HJBI equation.
\begin{definition}
\label{def: viscosity_sol}
A function $w\in C([0,1])$ is called a viscosity solution of \eqref{eq: HJBI_general} if it is both a \textit{viscosity subsolution} and a \textit{viscosity supersolution} that are respectively equivalent to:
\begin{subequations}
    \begin{align}
        \label{eq: subsol}
        &-w(p)+\sup_{\mu \in [0,1]} \inf_{h\in \BR} \left\{g(p,\mu,h) +\mathcal{K}(p,\xi_1,\xi_2,\mu,h)\right\}\leq 0, ~\forall (\xi_1,\xi_2) \in D_+w(p),\\
        \label{eq: supersol}
        &-w(p)+\sup_{\mu \in [0,1]} \inf_{h\in \BR} \left\{g(p,\mu,h) +\mathcal{K}(p,\xi_1,\xi_2,\mu,h)\right\}\geq 0, ~\forall (\xi_1,\xi_2) \in D_-w(p).
    \end{align}
\end{subequations}
\end{definition}
\begin{theorem}
\label{thm: valuvis}
Suppose $w \in C([0,1])$ is Lipschitz and a viscosity solution to the following HJBI equation: 
\begin{equation}
\label{eq: HJBI}
	w(p) = \sup_{\mu \in [0,1]} \inf_{h\in \BR} \left\{(1-\mu)r+\mu m(p)+\sigma\sqrt{\mu} h+\frac{\alpha}{2\delta}h^2 +\frac{\mu}{2\delta}\Phi(p) w''(p)\right\},
\end{equation}
where $\Phi(p):= \sigma^{-2}(\bar{\theta}-\underline{\theta})^2p^2(1-p)^2$. Then, $w$ equals $v$, the value function in \eqref{eq: supinfdef}. In the equilibrium, the worst-case density generator is $h^* =-\alpha^{-1}\sigma \delta \sqrt{\mu^*}$, where $\mu^*$ is the DM's best response in
\begin{equation}
\label{eq: hjb}
	w(p)=\sup_{\mu \in [0,1]}\left\{(1-\mu)r+\mu m(p)-\frac{\sigma^2 \delta}{2\alpha}\mu+\frac{\mu}{2\delta}\Phi(p)w''(p)\right\}\footnote{This equation should also be interpreted in the viscosity sense, by dropping the \textsf{infimum} in defnition \ref{def: viscosity_sol}.}.	
\end{equation}
\end{theorem}
As stated in previous theorem, on the equilibrium path of the game, DM knows the best response of the nature, that is $h(\mu) = -\alpha^{-1}\sigma \delta \sqrt{\mu}$. Therefore, her posterior process follows that of \eqref{eq: beliefmg} for the prescribed $h(\mu)$. Importantly, this means at the equilibrium the DM is no longer concerned about all possible distributions of past shocks. The one that has been picked by the nature is known to the DM on the equilibrium path, which gives rise to the unique law of motion for the posterior belief. Note that, this does not mean that ambiguity is mitigated on the equilibrium path. However, it simply means that similar to the static decision making, where the ambiguity averse agent first perceives the worst case distribution from her multiple prior set, and then responds back, here also she forms her belief and react based on the worst case distribution choice by the nature. Henceforth, by $p$ in \eqref{eq: hjb} and in the rest of the paper we mean the equilibrium posterior value, or often for brevity is simply referred as \textit{belief}. 

Note that the \textit{rhs} of \eqref{eq: hjb} is linear in $\mu$. This is in part due to the effect of $\sqrt{\mu}$ as the volatility term in the ambiguous arm. Consequently, the DM's optimal strategy at every point in time is to either \textit{explore} the ambiguous arm or \textit{exploit} the safe arm\footnote{The trade-off between exploration vs. exploitation has studied in different context. For one we can point to \cite{manso2011motivating} that explains such a trade-off for the financial incentives in entrepreneurship.}. As a result, the DM's value function satisfies the following variational relation:
\begin{equation}
\label{eq: varhjb}
	v(p) =\max \left\{r,m(p)-\frac{\sigma^2 \delta}{2\alpha}+\frac{1}{2\delta}\Phi(p)v''(p)\right\}	
\end{equation}
In the economic terms, $r$ is the DM's reservation value, which can always be achieved regardless of her experimentation strategy. The term $m(p)$ is the expected rate of return from pulling the second arm when the current belief on $\theta$ is $p$. The important term in expression \eqref{eq: varhjb} is $\sigma^2 \delta / 2\alpha$, which we call it \textit{ambiguity cost}. Higher ambiguity aversion, translated to lower $\alpha$, implies higher incurred cost upon pulling the ambiguous arm. Lastly, $\frac{1}{2\delta}\Phi(p) v''(p)$ is the continuation payoff that the DM could expect by holding on to the second arm. We postpone a more elaborate set of analytical results on the value function to the next subsection and instead present the intuition behind the DM's optimal strategy.
\begin{lemma}
The DM's optimal allocation choice with ambiguity aversion $\alpha$ admits the following representation:
\begin{align}
\label{eq: optmarkovstr}
		\mu^*(p)=\begin{cases}
		1 & \text{  if  }  \frac{1}{2\delta}\Phi(p)v''(p)-\frac{\sigma^2\delta}{2\alpha} > r - m(p)\\
		\in [0,1] & \text{  if   } \frac{1}{2\delta}\Phi(p)v''(p)-\frac{\sigma^2\delta}{2\alpha} = r - m(p)\\
		0  & \text{   otherwise   }
		\end{cases}
\end{align}
\end{lemma}
This result is the analogue of lemma 4 in \cite{bolton1999strategic} tailored to capture the ambiguity aversion. One shall think of $r-m(p)$ as the opportunity cost of experimentation that the DM incurs by not choosing the safe arm. Therefore, she only selects the second project when the continuation value of experimentation adjusted by the ambiguity price exceeds its opportunity cost. Particularly, whenever the two values match, the DM can pursue a \textit{mixed strategy}, in that she can allocate her resources between two arms in any arbitrary proportions. However, the Lebesgue measure of the time duration on which she chooses the mixed strategy is zero, precisely because $p$ follows a diffusion process and the middle case in \eqref{eq: optmarkovstr} never happens $\d \BP \times \mathsf{Leb}$-a.e. The ambiguity aversion essentially creates a situation in that the DM thinks that upon the continuation she will have to face with the most destructive types of shock distribution, and this already lowers the value of experimentation. Importantly, this loss is independent of the current belief level, and shall be viewed as a fixed cost that ambiguity averse agent must be compensated for to undertake the second project.
\section{Properties of the value function and comparative statics}
\label{sec: valuefunction}

In this section we propose closed-form expression for the value function and present sharp comparative statics with respect to ambiguity aversion index $\alpha$. 
\begin{theorem}
\label{thm: vconv}
On the equilibrium path the DM's follows a cut-off experimentation strategy. In particular, there exists $\bar{p} \in [0,1]$ such she selects the safe arm if and only if her posterior belief drops below $\bar{p}$. Further, the value function $v$ is convex on $[0,1]$.
\end{theorem}

A substantive result of convexity is that even in the presence of ambiguity aversion the marginal value of \textit{good news} about the second project is increasing.

Next, we want to find a closed-form expression for the value function and particularly the cut-off probability $\bar{p}$. For this we make a technical assumption that turns out to be necessary and sufficient for existence of $\bar{p}$ in $(0,1)$. Namely, we exclude the case $\bar{p}=0$ where DM always pulls the second arm, and $\bar{p}=1$ where she never does.
\begin{assumption}
\label{ass: etaass}
Define $\eta:= \frac{r-\underline{\theta}}{\overline{\theta}-\underline{\theta}}+\frac{\sigma^2 \delta }{2\alpha(\overline{\theta}-\underline{\theta})}$. Then we assume $\eta < 1$.
\end{assumption}
As becomes clear later, one can think of $\eta$ as a lower bound on $\bar{p}$. Therefore $\eta >1$ essentially means that DM never selects the ambiguous arm. This is due to a combination of two forces, namely a large ratio of safe to ambiguous return --- that is the first term in $\eta$ --- and high normalized ambiguity cost --- that is the second term in $\eta$ --- which prevents the DM from exploring the second arm. Assumption \ref{ass: etaass} not only ensures that $\bar{p}<1$, but as it will turn out it implies $\bar{p}>0$. Having made this assumption, on \textit{exploration region} $(\bar{p},1]$ the following differential equation holds:
\begin{equation}
	v(p) = m(p)-\frac{\sigma^2 \delta}{2\alpha}+\frac{1}{2\delta}\Phi(p)v''(p)
\end{equation}
That has a general solution form\footnote{\cite{polyanin2017handbook} page 547.}
\begin{equation}
	 v(p)=m(p)-\frac{\sigma^2\delta}{2\alpha}+cp^{1-\lambda}(1-p)^\lambda, ~~ \text{on} ~ p\in (\bar{p},1].	
\end{equation}
Here $c$ is a constant determined from the boundary condition and $\lambda=\frac{1+\sqrt{1+4\delta \varphi^{-2}}}{2}$, where $\varphi:=(\bar{\theta}-\underline{\theta})/\sigma \sqrt{2}$.
The \textit{value-matching} (or equivalently \textit{no-arbitrage}) condition implies that the DM should be indifferent between choosing any of the two arms at $p=\bar{p}$. Therefore, $v(\bar{p})=r$ that yields to
\begin{equation}
\label{eq: genform}
	v(p) = m(p)-\frac{\sigma^2\delta}{2\alpha}+\left(r-m(\bar{p})+\frac{\sigma^2\delta}{2\alpha}\right)\frac{p^{1-\lambda}(1-p)^\lambda}{\bar{p}^{1-\lambda}(1-\bar{p})^\lambda},\quad \forall p \in [\bar{p},1]. 
\end{equation}
The DM faces a \textit{free-boundary} problem, namely she needs to find the optimal cut-off $\bar{p}$. For that we need to apply the \textit{smooth-pasting}\footnote{\cite{dixit2013art}.} condition that imposes the continuity of directional derivatives at $\bar{p}$, i.e $v'(\bar{p}^-) = v'(\bar{p}^+)$. Assumption \ref{ass: etaass} with some amount of algebra yields to the following expression for the cut-off probability:
\begin{equation}
\label{eq: pbar}
	\bar{p} = \frac{(\lambda-1)\eta}{\lambda-\eta}
\end{equation}
It is positive because $\eta < 1 \leq \lambda$, and is less than one again because $\eta<1$. This observation now supports making assumption \ref{ass: etaass}.

\begin{remark}
The value function in \eqref{eq: genform} with the prescribed $\bar{p}$ is continuous, increasing and convex. Therefore, its maximum derivative is attained at $p=1$, that is bounded above because $\lambda>1$, thereby satisfying the Lipschitz continuity. Hence, $v$ owns all the properties of the verification theorem \ref{thm: valuvis}.
\end{remark}

\noindent\textbf{Some comparative statics.} The cut-off value is lower-bounded by 
$\eta$. Further, it is increasing in $\eta$. Expression \eqref{eq: pbar} provides us with a sharp characterization of the cut-off value, and one could perform a number of comparative statics on $\bar{p}$ with respect to the parameters of the model. Here, we only point to two interesting ones. First, and more important is the effect of ambiguity on cut-off value. As DM becomes more ambiguity averse, namely as $\alpha$ becomes smaller, the value of $\bar{p}$ increases unambiguously. This confirms our intuition that a more ambiguity averse DM is more conservative and explores less. Expression \eqref{eq: pbar} offers a fine indicator on the \textit{extent} of this under-exploration. The second channel is the effect of $\ovl{\theta}-\unl{\theta}$, that represents the \textit{range} of possible return rates under the second arm. As this range shrinks to zero, the ambiguity cost is amplified more intensely, and DM will have less incentive to pick the second project.

As a last note in this section we point out to a concern on the entangled effects of $\sigma$ and $\alpha$. One might wonder that what we refer as the ambiguity aversion parameter, i.e $\alpha^{-1}$, can be dissolved in volatility $\sigma$, and thus can never be identified separately even with infinite amount of data. However, this is not true, as we can offer an identification scheme that disentangles $\alpha$ from $\sigma$. Suppose that all other parameters are identified, namely $r,\delta$ and $\{\bar{\theta},\underline{\theta}\}$. Then, a continuous stream of agent's belief process would let us to compute the quadratic variation $\angbrac{p}{p}=(\bar{\theta}-\underline{\theta})^2p^2(1-p)^2/\sigma^2$ from \eqref{eq: beliefmg}. Further, by spotting the point where she stops the exploration and pulls the safe arm we can back out $\bar{p}$. These two equations can lead us to uniquely identify $\sigma$ and $\alpha$.

\section{Value of unambiguous information}
\label{sec:value_of_unambiguous_information}
In this section we aim to study the value of information with respect to which the DM holds no ambiguity. Practically, one can think of a scenario in which the experimentation unit hires an expert to continuously provide her opinion about the \textit{true} rate of return of the ambiguous arm. Some questions naturally arise in this context. For example what is the fair price of such service? Or, how much must the expert be compensated for providing such information? When should the experimentation unit who faces ambiguity hire this expert?

To answer such questions, let $x_t$ be the information that the expert releases at time $t$ about $\theta$, which in its simplest case can be thought as the noisy signal of $\theta$, namely:
\begin{equation}
\label{eq: unambiguous_signa}
	\d x_t = \theta \d t+\gamma \d W_t
\end{equation}
In this expression $W$ is a $\mb{F}$-Brownian motion under the benchmark measure $\BP$ and is independent of $B$ and $\theta$. Further, $\gamma$ is the constant volatility that represents the level of DM's confidence in the expert's information. Therefore, the DM can use this signal in addition to the second arm's payoff process to update her belief about $\theta$. Obviously, this new source of information improves the precision of the filtering process, in the sense that it lowers the conditional variance of estimated $\theta$ at every point in time. The law of motion for the new posterior process with the presence of unambiguous information source follows the logic of lemma \ref{lem: optfilt}:
\begin{equation}
	\d p_t^h = p_t^h(1-p_t^h)\left(\bar{\theta}-\underline{\theta}\right)\left[\frac{\sqrt{\mu_t}}{\sigma}\d \bar{B}^h_t+\frac{1}{\gamma}\d \overline{W}_t\right]
\end{equation}
Here $\bar{B}$ and $\overline{W}$ are independent $\mb{F}^{y_2,\mu,h,x}$-Brownian motions under $\BP^h$.
Now we can state the counterpart of theorem \ref{thm: valuvis} in this case, however its proof is easier as the candidate solution belongs to the space of $C^2([0,1])$ thus we do not need the viscosity solution concept. This is owed to the fact that the diffusion coefficient for $\overline{W}$ is independent of $\mu$, thereby relaxing the degeneracy that appears when $\mu=0$. As a result of restriction to the space $C^2([0,1])$, Ito's lemma can be applied directly on the candidate value function and one can apply the idea of the proof in theorem \ref{thm: valuvis}, bypassing the steps dealing with viscosity super(sub)-solution and replacing them with Ito's rule.
\begin{proposition}
\label{prop: infovaluvis}
Suppose $\tilde{v} \in C^2([0,1])$ is the unique solution to the following HJBI equation:
\begin{equation}
    \tilde{v}(p)=\sup_{\mu \in [0,1]} \inf_{h \in \BR}\left\{(1-\mu)r+\mu m(p)+\sqrt{\mu}\sigma h+\frac{\alpha}{2\delta}h^2+\frac{1}{2\delta}\left(\mu\Phi(p;\sigma)+\Phi(p;\gamma)\right)\tilde{v}''(p)\right\}
\end{equation}
In that $\Phi(p;s) := \frac{(\bar{\theta}-\underline{\theta})^2}{s^2}p^2(1-p)^2$. Then, $\tilde{v}$ is indeed the value function in presence of unambiguous information $x$. In the equilibrium, the worst-case density generator is $h^* =-\alpha^{-1}\sigma \delta \sqrt{\mu^*}$, where $\mu^*$ is the DM's best response solving:
\begin{equation}
\label{eq: infohjb}
	\tilde{v}(p)=\sup_{\mu \in [0,1]}\left\{(1-\mu)r+\mu m(p)-\frac{\sigma^2 \delta}{2\alpha}\mu+\frac{1}{2\delta}\left(\mu\Phi(p;\sigma)+\Phi(p;\gamma)\right)\tilde{v}''(p)\right\}	
\end{equation}
\end{proposition}
Similar to the case with no source of unambiguous information, one can show that the value function is non-decreasing in $p$ and there is a cut-off rule for the optimal experimentation strategy. Let us denote the new cut-off in the presence of unambiguous information with $\tilde{p}$. Then, the value function satisfies the following relation:
\begin{align}
\label{eq: double_HJB}
		\tilde{v}(p)=\left\{ \begin{array}{ll}
		r+\delta^{-1}\varphi(\gamma)^2p^2(1-p)^2\tilde{v}''(p) & p < \tilde{p}\\
		m(p)-\frac{\sigma^2\delta}{2\alpha}+\delta^{-1}\left(\varphi(\sigma)^2+\varphi(\gamma)^2\right)\tilde{v}''(p) & p > \tilde{p}
		\end{array}\right.	
\end{align}
In that we define $\varphi(s) = \left(\bar{\theta}-\underline{\theta}\right)/s\sqrt{2}$, where $s \in \{\sigma,\gamma\}$. The top term in \eqref{eq: double_HJB} relates to the region where DM selects the safe arm. Importantly, on this region her payoff is no longer $r$, but has a continuation component that arises from the free information $x$. In the case without this source, once the DM switches to the safe arm, she will never have the chance to acquire information about $\theta$, thereby her payoff will stuck at $r$ forever. However, in the current situation, the news about $\theta$ can still be flowing without DM pulling the second arm, and in the case of \textit{good} news, she would expect to switch back to the second arm. This effect creates the continuation incentives for the DM on the region $(0,\tilde{p})$. At $\tilde{p}$ the continuity condition must hold so any of the two regions in \eqref{eq: double_HJB} could be enclosed. The solution to this piece-wise ordinary differential equation is
\begin{align}
		\tilde{v}(p)=\left\{ \begin{array}{ll}
		r+c_1p^{\lambda_1}(1-p)^{1-\lambda_1} & p < \tilde{p}\\
		m(p)-\frac{\sigma^2 \delta}{2\alpha}+c_2 p^{1-\lambda_2}(1-p)^{\lambda_2} & p > \tilde{p},
		\end{array}\right.	
\end{align}
where $\lambda_1 =\frac{1+\sqrt{1+4\delta \varphi(\gamma)^{-2}}}{2}$ and $\lambda_2=\frac{1+\sqrt{1+4\delta\left(\varphi(\sigma)^2+\varphi(\gamma)^2\right)^{-1} }}{2}$. There are essentially three parameters to be determined, i.e $(c_1,c_2,\tilde{p})$. The optimal choice of DM is to select these constants so that the three conditions, namely \textit{value-matching} (continuity), \textit{smooth-pasting} (continuity of first derivative) and \textit{super-contact} (continuity of second derivative) hold together. The derivations for this are presented in \ref{subsec: constants_der}. It turns out the new cut-off probability under unambiguous information source is
\begin{equation}
\label{eq: ptilde}
	\tilde{p}=\frac{(\Lambda-1)\eta}{\Lambda-\eta},\quad \text{for }\Lambda:=1+\lambda_1\frac{\sigma^2}{\gamma^2}+(\lambda_2-1)\left(1+\frac{\sigma^2}{\gamma^2}\right).
\end{equation}

\begin{proposition}
\label{prop: Compar_cutoff}
The experimentation cut-off rises unambiguously when there is an unambiguous information source, namely $\tilde{p}\geq \bar{p}$ for all combinations of the variables in the model.
\end{proposition}
The content behind this proposition is that the unambiguous source of information in effect raises the bar for exploration, that in turn means DM demands more confidence for selecting the second project. This is very much due to the free information that DM can acquire about $\theta$ without pulling the ambiguous arm. In the standard case, the only way to learn about the quality of the second project is to spend some time exploring that. Therefore, the DM is more willing to sacrifice the certain payoff of the first project to learn about the second one, whereas in the current case she can wait longer for the good news (and exploit the first arm meanwhile) to choose the second arm. In this spirit, as depicted in figure \ref{fig:fig1} the cut-off value rises unambiguously due to the provision of the new information source (i.e $\tilde{p}>\bar{p}$). Also it shows that in both environments the exploration threshold falls as the DM becomes less ambiguity averse, meaning larger values of $\alpha$.
\begin{figure}
    \centering
    \includegraphics[scale=0.7]{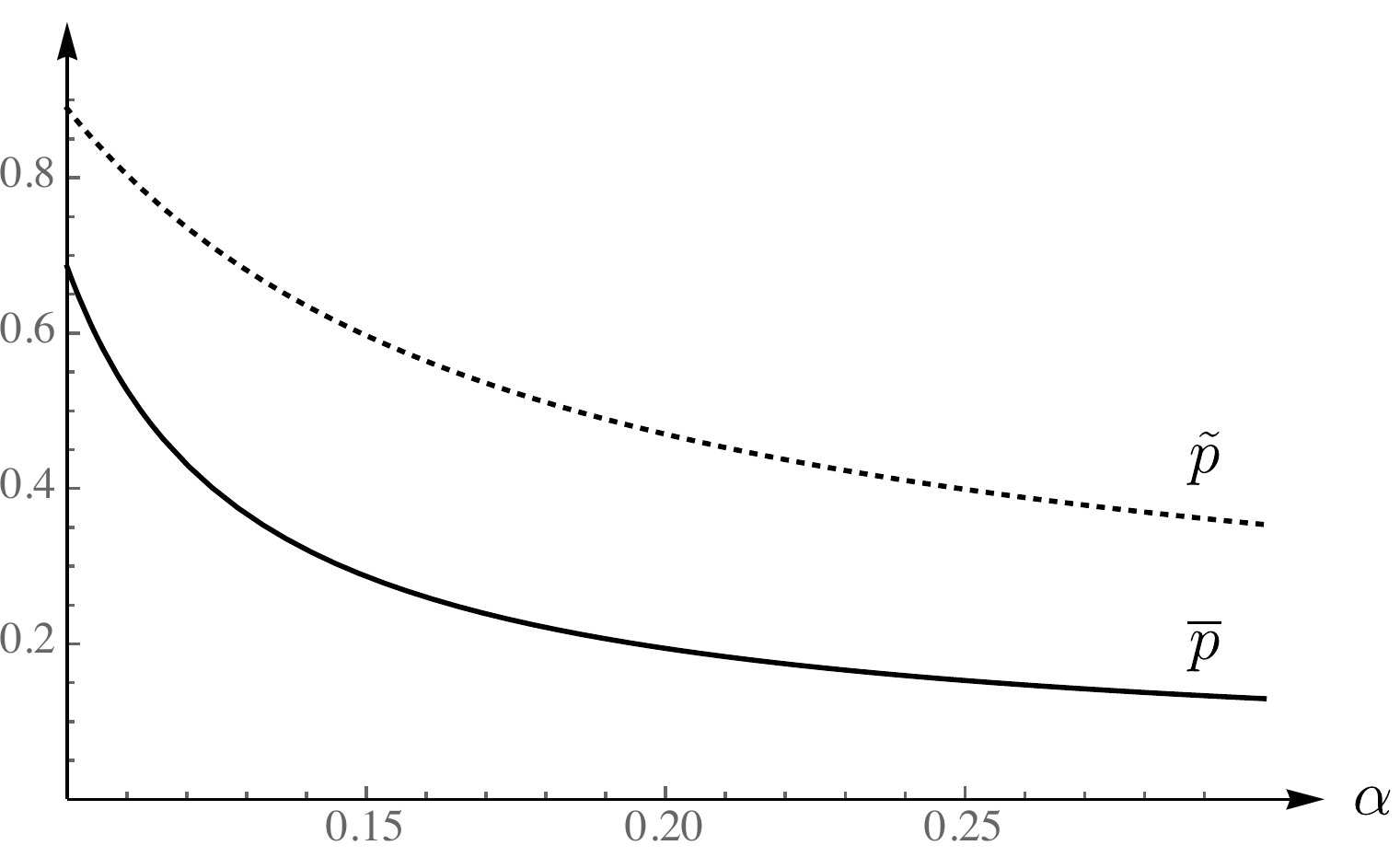}
    \caption{Cut-off values}
    \medskip
    \small
    $\left[r=0.2,\unl{\theta}=0, \ovl{\theta}=1, \delta=0.9,\sigma=0.4, \gamma=0.3\right]$
    \label{fig:fig2}
\end{figure}
One can think of a situation where the provider of this new source of information is strategic and can charge the DM for the service. Then naturally the maximum price that she can charge is $\tilde{v}(p)-v(p)$, which corresponds to extracting all the surplus from the DM. From the social welfare standpoint the $p$ that maximizes the surplus shall be treated as a benchmark for decision to hire the expert. We refer to $\tilde{v}(p)-v(p)$ as the created surplus due the expert opinion. It is obviously positive and continuous everywhere, and is increasing over $[0,\bar{p}]$. Also as $p \to 1$ it decays to zero faster than $(1-p)^{\lambda_2}$. Therefore, the maximum created surplus occurs at a \textit{moderate} belief value $p^*$, where $p^* > \bar{p}$ but is not also very close to one. Figure \ref{fig:fig2} presents both value functions, and the created surplus. In that the blue segment of each curve points to the region where the DM pulls the safe arm. We end this section with a remark about the most efficient time to hire an expert.
\begin{remark}
The above analysis implies that it is most beneficial for the experimentation unit to hire an expert when otherwise they would select the ambiguous arm in spite of strong enough evidence and belief.
\end{remark} 
\begin{figure}
    \centering
    \includegraphics[scale=0.7]{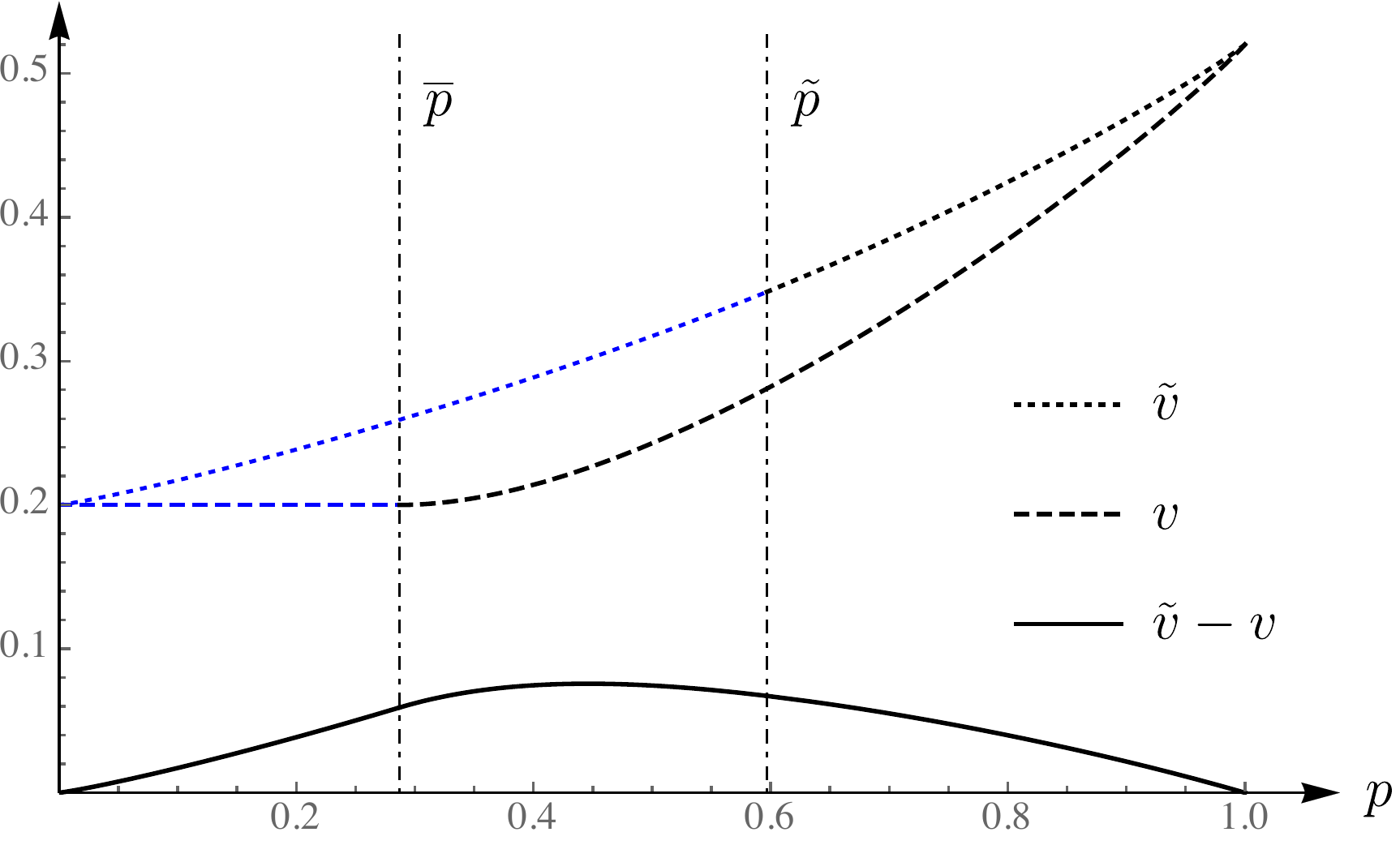}
    \caption{Created surplus and value functions}
    \medskip
    \small
    $\left[r=0.2,\unl{\theta}=0, \ovl{\theta}=1, \delta=0.9,\sigma=0.4, \gamma=0.3, \alpha=0.14\right]$
    \label{fig:fig1}
\end{figure}

\section{Concluding remarks} 
\label{sec:concluding remarks}
How does a decision maker who is \textit{uncertain} about the payoff distribution of two alternative choices operate the dynamics of experimentation? Understanding how an ambiguity averse agent values a project and determining the \textit{price} of ambiguity are particularly important when the experimentation task is delegated to such agent. In this paper, we develop a dynamic decision making framework that offers closed-form characterizations for the agent's optimal strategy as well as her valuation. Specifically, we assumed the DM has Multiplier preferences, that consists of two components. The discounted expected future return from both arms, and a penalty term that captures the extent of perturbation of probability specification relative to the benchmark model. We framed the decision making environment as a two-player differential game that DM plays against the nature, and found a closed-form expression for DM's value function in terms of her belief. Also, we have shown that in the equilibrium her optimal strategy is to select the safe arm of the project whenever her belief drops below a certain threshold, the value of which is controlled by all the parameters of the model and specifically the ambiguity aversion index. Our analysis offers sharp results on how much an ambiguity averse DM must be compensated to act as if she is not subject to ambiguity. In particular, one can send $\alpha \to \infty$ in the results of section \ref{sec: valuefunction} to predict the behavior of an ambiguity neutral agent. Finally, we explored the effect of an unambiguous constantly flowing information source in the dynamics of experimentation. It turned out that the exploration cut-off rises as a result of such provision, namely the DM waits longer to receive good news about the ambiguous arm of the project. We investigated the generated surplus due to this additional source and offered policy analysis on the efficient time to recruit an external expert to guide the experimentation process.

\newpage
\appendix
\section{Proofs}
\label{sec: proofs}
\subsection{Proof of lemma \ref{lem: wellentr}}
For every finite $T$ and $h\in \mc{H}$ the integral can be simplified as:
\begin{equation}
	\begin{split}
		\delta \int_0^T e^{-\delta t} H\left(\BP^{h}_t ; \BP_t\right)\d t  &= \delta \int_0^T e^{-\delta t} \BE^{h}\left[\log L_t^{h}\right]\d t\\
		&=\delta \int_0^T e^{-\delta t}\BE^{h}\left[\left(h \cdot B\right)_t-\frac{1}{2}\left(h^2 \cdot \imath\right)_t\right]\d t\\
		&=\delta \int_0^T e^{-\delta t}\BE^{h}\left[\left(h \cdot B^{h}\right)_t+\frac{1}{2}\left(h^2 \cdot \imath\right)_t\right]\d t
	\end{split}
\end{equation}
Since $\{B^{h}_t: t\leq T\}$ is $\BP^{h}_T$-Brownian motion and $h$ is bounded, the first term in above has zero expectation, leaving us only with the second term, for which integration by part yields
\begin{equation}
\label{appeq: finitent}
	\BE^{h}\left[\frac{\delta}{2} \int_0^T e^{-\delta t}\left(h^2 \cdot \imath\right)_t \d t\right] = \BE^{h}\left[-\frac{1}{2}e^{-\delta T}(h^2 \cdot \imath)_T+\frac{1}{2}\int_0^T e^{-\delta t}h^2_t \d t\right].
\end{equation}
The first term inside expectation is uniformly bounded over $\Omega \times \BR_+$ and goes to zero in a point-wise sense as $T \to \infty$. Therefore,
\begin{equation}
	H\left(\BP^h; \BP\right)= \frac{1}{2}\lim_{T \to \infty}\BE^{h}\left[\int_0^T e^{-\delta t}h^2_t \d t\right] = \frac{1}{2}\BE^{h}\left[\int_0^\infty e^{-\delta t}h^2_t \d t\right],
\end{equation}
where in the last relation we used the monotone convergence theorem. The limit is finite due to the boundedness of $h \in \mc{H}$. It is worthwhile to point out that for integrals with finite upper limit $T$ we appeal to $\BP^{h}_T$, and for the infinite time integral we use $\BP^{h}$. This replacement does not cause any problem because of the consistency of $\{\BP^{h}_T: T \in \BR_+\}$ with $\BP^{h}$ as explained in item \ref{item: prob_cons} of subsection \ref{sub:weak_formulation}.\qed

\subsection{Proof of proposition \ref{prop: payoffrep}}
Let us define 
\begin{equation}
\label{appeq: finiteV}
	V^T(p;\mu,h):=\BE^h\left[\delta \int_0^T e^{-\delta t} \left(\d y_{1,t}+\d y_{2,t}+\alpha H\left(P_t^h;P_t\right)\d t\right)\right].	
\end{equation}
For the first two components of \eqref{appeq: finiteV}, one just need to recall that over every finite interval $[0,T]$, the pair $\left\{B^{h}_t, \mc{F}_t: t\leq T\right\}$ is a Brownian motion under $\BP^h$. Consequently, stochastic integrals of bounded processes with respect to that are martingales and hence average out to zero. Using equation \eqref{eq: filteredy2} yields to
\begin{equation}
\label{appeq: y1y2h}
	\begin{split}
		\BE^h\left[\delta \int_0^Te^{-\delta t} \left(\d y_{1,t}+\d y_{2,t}\right)\right]&=\BE^h\left[\delta \int_0^T e^{-\delta t}\left\{\left((1-\mu_t)r+\mu_tm(p_t^h)+\sigma \sqrt{\mu_t}h_t\right)\d t+ \sigma \sqrt{\mu_t}\ d \bar{B}^h_t\right\}\right]\\
		&=\BE^h\left[\delta \int_0^T e^{-\delta t}\left((1-\mu_t)r+\mu_tm(p_t^h)+\sigma \sqrt{\mu_t}h_t\right)\d t \right]
	\end{split}	
\end{equation}
The entropy component of the integrand in \eqref{appeq: finiteV} has already been analyzed in the proof of lemma \ref{lem: wellentr} and specifically in equation \eqref{appeq: finitent}. Combining that analysis with \eqref{appeq: y1y2h} leads to
\begin{equation}
\label{appeq: VTsimplified}
	\begin{split}
		V^T(p;\mu,h)&=\BE^{h}\left[\delta \int_0^T e^{-\delta t}\left((1-\mu_t)r+ \mu_tm(p_t^h) +\sigma \sqrt{\mu_t}h_t+\frac{\alpha}{2\delta}h^2_t\right)\d t \right]\\
		&-\frac{1}{2}e^{-\delta T}\BE^h\left[(h^2 \cdot \imath)_T\right]
	\end{split}	
\end{equation}
Since $h \in \mc{H}$ is bounded, the second term in \eqref{appeq: VTsimplified} vanishes as $T \to \infty$. For the first term in \eqref{appeq: VTsimplified} we apply the dominated convergence theorem and use the fact that $\left\{\BP_T^{h}: T \in \BR_+\right\}$ is consistent with $\BP^{h} \in  \Delta\left(\Omega,\mc{F}_\infty\right)$\footnote{$\Delta\left(\Omega,\mc{F}\right)$ denotes the set of all probability measures on the measure space $\left(\Omega,\mc{F}\right)$.} --- as explained in \ref{item: prob_cons} of subsection \ref{sub:weak_formulation} --- while sending $T \to \infty$. This concludes the proof of $\lim_{T\to \infty} V^T(p;\mu,h) = V(p;\mu,h)$, thereby leading to \eqref{eq: payoffrep}.\qed

\subsection{Proof of theorem \ref{thm: valuvis}}
For the proof of this theorem we need the following lemma proved in \cite{zhou1997stochastic} using the Mollification method.
\begin{lemma}
\label{lem: mollif}
Let $w\in C([0,1])$ be a given Lipschitz function. For every $x_0 \in [0,1)$, with $(\xi_1,\xi_2) \in D_+w(x_0)$ (resp. $(\xi_1,\xi_2)\in D_-w(x_0)$), there exists a twice-continuously differentiable function $\psi\in C^2([0,1])$, that satisfies:
\begin{enumerate}[label=(\roman*)]
    \item $\psi(x_0)=w(x_0)$ and $\psi(x)> w(x)$ (resp. $\psi(x)< w(x)$) everywhere else.
    \item $\psi'(x_0)=\xi_1$ and $\psi''(x_0)=\xi_2$.
\end{enumerate}
\end{lemma}
To continue the proof of the theorem assume $\exists w \in C([0,1])$ satisfying all presumptions of the theorem. We also use the following notation throughout the proof:
\begin{equation}
    g(p,\mu,h):= (1-\mu)r+\mu m(p)+\sigma\sqrt{\mu} h+\frac{\alpha}{2\delta}h^2
\end{equation}
Recall that $\{p_t^h\}$ follows the diffusion process $\d p_t^h=\sqrt{\mu_t\Phi(p_t^h)}\d \bar{B}_t^h$, where $\bar{B}^h$ is a $\mc{G}$-Brownian motion under $\BP^h$ over any finite horizon. Suppose at $t=0$, $p_t^h=p$, and $\left(\partial_-w(p),\partial^2_-w(p)\right) \in D_-w(p)$, then from lemma \ref{lem: mollif} one can find $\psi \in C^2([0,1])$ such that $\psi(p)=w(p)$, $\psi(x)< w(x)$ elsewhere, $\psi'(p)=\partial_-w(p)$ and $\psi''(p)=\partial^2_-w(p)$. Then, for every $t>0$:
\begin{equation}
    \begin{split}
        e^{-\delta t}w(p_{t}^h)-w(p) &\geq e^{-\delta t} \psi(p_t^h)-\psi(p)\\
        &=\int_0^t e^{-\delta s}\left[ \sqrt{\mu_s\Phi(p_s^h)} \psi'(p_s^h) \d \bar{B}_s^h+\left(\frac{1}{2}\mu_s \Phi(p_s^h)\psi''(p_s^h)-\delta \psi(p_s^h)\right)\d s\right]
    \end{split}
\end{equation}
Therefore, taking the expectation w.r.t $\BP^h$ on both sides and using the martingale property of $\bar{B}^h$ lead to:
\begin{equation}
    \begin{split}
        \frac{1}{t}\left( \BE^h\left[e^{-\delta t}w(p_t^h)\right]-w(p)\right)\geq \frac{1}{t}\int_0^t e^{-\delta s}\BE^h\left[\frac{1}{2}\mu_s \Phi(p_s^h)\psi''(p_s^h)-\delta \psi(p_s^h)\right]\d s
    \end{split}
\end{equation}
Since $p^h_s\to p$, $\BP^h$-almost surely as $s\to 0$, and hence in distribution, then taking the limit on both sides as $t\to 0$ yields:
\begin{equation}
    \begin{split}
        \liminf_{t\to 0}\frac{1}{t}\left( \BE^h\left[e^{-\delta t}w(p_t^h)\right]-w(p)\right)&\geq \frac{1}{2}\mu \Phi(p)\psi''(p)-\delta \psi(p)\\
        &=\frac{1}{2}\mu \Phi(p)\partial^2_- w(p)-\delta w(p)
    \end{split}
\end{equation}
Let $\mu=\mu^*$ in the above expression. Since $w$ is the viscosity solution for \eqref{eq: HJBI}, then from the \textit{supersolution} property \eqref{eq: supersol} it holds that $\inf_h\left\{g(p,\mu^*,h)+\frac{\mu^*}{2\delta}\Phi(p)\partial^2_-w(p)-w(p)\right\}\geq 0$, therefore for every $h$:
\begin{equation}
    \frac{1}{2}\mu^* \Phi(p)\partial^2_- w(p)-\delta w(p) \geq -\delta g(p,\mu^*,h).
\end{equation}
Combining the last two equations amounts to
\begin{equation}
\label{eq: Ito_replace}
    \liminf_{t\to 0}\frac{1}{t}\left( \BE^h\left[e^{-\delta t}w(p_t^h)\right]-w(p)\right) \geq -\delta g(p,\mu^*,h).
\end{equation}
This is a fundamental implication that one would have obtained much easier under Ito's lemma if $w$ was twice continuously differentiable. 

Next, for every $t>0$
\begin{equation}
\label{eq: Yosida_recipe}
    \begin{split}
        \BE^h\left[e^{-\delta t}w(p_t^h)\right]-w(p) &= \lim_{\ve \to 0} \frac{1}{\ve}\left(\int_t^{t+\ve} \BE^h\left[e^{-\delta s}w(p_s^h)\right]\d s -\int_0^\ve \BE^h\left[e^{-\delta s}w(p_s^h)\right] \d s\right)\\
        &=\lim_{\ve \to 0} \frac{1}{\ve}\left(\int_\ve^{t+\ve} \BE^h\left[e^{-\delta s}w(p_s^h)\right]\d s -\int_0^t \BE^h\left[e^{-\delta s}w(p_s^h)\right] \d s\right)\\
        &=\lim_{\ve \to 0} \int_0^t \frac{\BE^h\left[e^{-\delta (s+\ve)}w(p_{s+\ve}^h)\right]-\BE^h\left[e^{-\delta s}w(p_s^h)\right]}{\ve}\d s\\
        &\geq \int_0^t \liminf_{\ve \to 0} \frac{\BE^h\left[e^{-\delta (s+\ve)}w(p_{s+\ve}^h)\right]-\BE^h\left[e^{-\delta s}w(p_s^h)\right]}{\ve}\d s\\
        &\geq -\delta \int_0^t \BE^h\left[e^{-\delta s}g(p_s^h,\mu^*_s,h_s)\right]\d s,
    \end{split}
\end{equation}
where in the second last inequality we used the Fatou's lemma, given that $w$ is bounded from above, and in the last inequality we used the inequality \eqref{eq: Ito_replace}. Rearranging the above terms implies that for every $t>0$
\begin{equation}
    w(p) \leq  \BE^h \left[\int_0^t \delta e^{-\delta s}g(p_s^h,\mu^*_s,h_s)\d s\right]+\BE^h\left[e^{-\delta t}w(p_t^h)\right]
\end{equation}
Since every $h \in \mc{H}$ is assumed bounded, then one can use dominated convergence theorem and send $t\to \infty$ to obtain
\begin{equation}
    w(p) \leq \BE^h \left[\int_0^\infty \delta e^{-\delta s}g(p_s^h,\mu^*_s,h_s)\d s\right].
\end{equation}
Taking the infimum over all $h \in \mc{H}$ thus implies
\begin{equation}
    w(p) \leq \inf_{h \in \mc{H}} \BE^h \left[\int_0^\infty \delta e^{-\delta s}g(p_s^h,\mu^*_s,h_s)\d s\right],
\end{equation}
and consequently
\begin{equation}
\label{eq: verif_upper}
    w(p) \leq \sup_{\mu}\inf_{h \in \mc{H}} \BE^h \left[\int_0^\infty \delta e^{-\delta s}g(p_s^h,\mu_s,h_s)\d s\right].
\end{equation}
For the reverse direction of the above inequality we shall use the superdifferentials of $w$ at $p$ and the viscosity subsolution inequality \eqref{eq: subsol}. Let $\left(\partial_+w(p),\partial^2_+ w(p)\right)\in D_+w(p)$.  Using an analogous argument as above we reach:
\begin{equation}
    \limsup_{t\to 0}\frac{1}{t}\left( \BE^h\left[e^{-\delta t}w(p_t^h)\right]-w(p)\right) \leq \frac{1}{2}\mu \Phi(p)\partial^2_+ w(p)-\delta w(p) 
\end{equation}
Choose an arbitrary $\mu \in \mc{U}$ and set $h=\tilde{h}=-\alpha^{-1}\sigma \delta \sqrt{\mu}$ ---  the point achieving the infimum in the HJBI equation. Because of the subsolution property of $w$, it holds that $g(p,\mu,\tilde{h})+\frac{\mu}{2\delta}\Phi(p)\partial^2_+w(p)-w(p) \leq 0$, therefore
\begin{equation}
    \limsup_{t\to 0}\frac{1}{t}\left( \BE^h\left[e^{-\delta t}w(p_t^h)\right]-w(p)\right) \leq -\delta g(p,\mu,\tilde{h}).
\end{equation}
Employing the same recipe of \eqref{eq: Yosida_recipe} -- this time with limsup instead of liminf in the Fatou's lemma and using the lower bound instead of upper bound on $w$ -- we get
\begin{equation}
    w(p) \geq   \BE^{\tilde{h}} \left[\int_0^t \delta e^{-\delta s}g(p_s^{\tilde{h}},\mu_s,\tilde{h}_s)\d s\right]+\BE^{\tilde{h}}\left[e^{-\delta t}w(p_t^{\tilde{h}})\right].
\end{equation}
Using the dominated convergence theorem to send $t\to \infty$ implies
\begin{equation}
    w(p) \geq \BE^{\tilde{h}} \left[\int_0^\infty \delta e^{-\delta s}g(p_s^{\tilde{h}},\mu_s,\tilde{h}_s)\d s\right] \Rightarrow w(p) \geq \inf_{h \in \mc{H}} \BE^h \left[\int_0^\infty \delta e^{-\delta s}g(p_s^h,\mu_s,h_s)\d s\right].
\end{equation}
Taking the supremum over all $\mu \in \mc{U}$ yields 
\begin{equation}
\label{eq: verif_lower}
    w(p) \geq \sup_{\mu}\inf_{h \in \mc{H}} \BE^h \left[\int_0^\infty \delta e^{-\delta s}g(p_s^h,\mu_s,h_s)\d s\right].
\end{equation}
Equations \eqref{eq: verif_upper} and \eqref{eq: verif_lower} together imply that $w=v$, that concludes the verification proof.\qed


\subsection{Proof of theorem \ref{thm: vconv}}
For the proof of this proposition we need few lemmas.
\begin{lemma}
For any $p \in (0,1)$ the value function is lower bounded by $\max\left\{r,m(p)-\frac{\sigma^2\delta}{2\alpha}\right\}$.
\end{lemma}
\begin{proof}\renewcommand\qedsymbol{} 
By replacing nature's best response $h = -\alpha^{-1}\sigma\delta \sqrt{\mu}$ in \eqref{eq: payoffrep}, one gets the following payoff representation:
\begin{equation}
\label{appeq: vsuprep}
 	v(p)=\sup_\mu \BE^\mu\left[\delta \int_0^\infty e^{-\delta t}\left((1-\mu_t)r + \mu_t m(p_t^\mu)-\mu_t\frac{\sigma^2\delta}{2\alpha}\right)\d t\right],
\end{equation}
where $\BE^\mu$ and $p_t^\mu$ are resp. the probability measure and the posterior probability obtained from $h = -\alpha^{-1}\sigma\delta \sqrt{\mu}$. Furthermore, using the local-martingale property of $m(p_t^h)$, we get the following inequality for every $\mc{G}_0$-measurable control process, i.e $\mu_t \in \mc{G}_0$ for all $t \in \BR_+$, and hence $\mu_t \equiv \mu$ (up to evanescence):
\begin{equation}
	\begin{split}
		v(p)  &= \sup_{\mu \in \mc{U}} \BE^\mu\left[\delta \int_0^\infty e^{-\delta t}\left((1-\mu_t)r + \mu_t m(p_t)-\mu_t\frac{\sigma^2\delta}{2\alpha}\right)\d t\right] \\
		&\geq \sup_{\mu \in \mc{G}_0}\left\{\BE^\mu\left[\delta \int_0^\infty e^{-\delta t}\left((1-\mu)r + \mu m(p_0)-\mu\frac{\sigma^2\delta}{2\alpha}\right)\d t\right]+ \BE^\mu\left[\delta \int_0^\infty e^{-\delta t} \mu x^{h}_t \d t\right]\right\}
	\end{split}	
\end{equation}
Here $x^h$ is the local-martingale part of $m(p^h)$ resulted from lemma \ref{eq: beliefmg}. Having set $\mu \in \mc{G}_0$, the expectation of the second term vanishes due to the $\BP^h$-local-martingale property of $x^h$ and using the dominated convergence theorem for approximating the infinite horizon integral with finite counterparts. This proves the lower bound on $v(p)$.$\parallel$
\end{proof}

\begin{lemma}
Let $\mc{S}_i$ be subset of $[0,1]$ where the DM optimally chooses the $i$-th project if $p \in \mc{S}_i$, where $i \in \{1,2\}$. Then the value function is convex restricted to each of these subsets.
\end{lemma}
\begin{proof}\renewcommand\qedsymbol{} 
On $\mc{S}_1$ the value function is identical to $r$, and hence is convex. On $\mc{S}_2$ the DM chooses the second arm and $\mu=1$, hence
\begin{equation}
	v(p) = m(p)-\frac{\sigma^2\delta}{2\alpha}+\frac{1}{2\delta} \Phi(p)v''(p),	
\end{equation}
which implies that
\begin{equation}
	\frac{1}{2\delta} \Phi(p)v''(p) =v(p)-m(p)+\frac{\sigma^2\delta}{2\alpha} \geq \left(m(p)-\frac{\sigma^2\delta}{2\alpha}\right)-m(p)+\frac{\sigma^2\delta}{2\alpha} =0.
\end{equation}
Therefore, $v''(p)\geq 0$ and hence the restriction of $v$ onto $\mc{S}_2$ is also convex.$\parallel$
\end{proof}
\begin{lemma}
The subsets $\mc{S}_1$ and $\mc{S}_2$ are connected subsets of $[0,1]$.
\end{lemma}
\begin{proof}\renewcommand\qedsymbol{} First note that $[0,1]=\mc{S}_1 \cup \mc{S}_2$, therefore the case of one subset being the empty set and the other being the whole unit interval trivially passes the lemma. Now assume both subsets are non-empty, and suppose $\mc{S}_1$ is not connected. Therefore, it must contain two disjoint open intervals, say $(a_1,b_1)$ and $(a_2,b_2)$, such that $b_1 < a_2$. This means that $[b_1,a_2]\subset \mc{S}_2$. The continuity must holds at the boundaries, namely $v(b_1)=v(a_2)=r$, otherwise there appears an \textit{arbitrage} opportunity for the DM and she could improve her strategy subsets, $\mc{S}_1$ and $\mc{S}_2$, so as to strictly be better off. Also, one can easily confirm from \eqref{appeq: vsuprep} that $v(\cdot)$ is a non-decreasing function in $p$. Since $v(\cdot)$ is always greater than or equal to $r$, then $v\equiv r$ on the entire $[b_1,a_2]$. This means essentially $[b_1,a_2] \subset \mc{S}_1$ that violates the initial assumption on $\mc{S}_1$. Therefore, $\mc{S}_1$ must be a connected subset of $[0,1]$. We use the proof by contradiction again to show $\mc{S}_2$ is connected as well. Suppose it is not, then it contains two disjoint open sets, say $(c_1,d_1)$ and $(c_2,d_2)$ such that $d_1 < c_2$. Note that at the boundary points the continuity must hold --- precisely to rule out the arbitrage ---  that means $v(d_1)=v(c_2)=r$. This means either $v \equiv r$ on $(c_1,d_1)$, which then one should include this interval in $\mc{S}_1$, or there exists some point $z \in (c_1,d_1)$ such that $v(z) > r$. This violates the non-decreasingness of $v$, and hence concludes the proof.$\parallel$
\end{proof}

The existence of cut-off strategy now falls out of the connectedness of $\mc{S}_1$ and $\mc{S}_2$ from previous lemma and monotonicity of $v(\cdot)$. It is thus left to prove the global convexity of $v(\cdot)$. For this denote the cut-off point by $\bar{p}$, and note that $\mc{S}_1=[0,\bar{p}]$ and $\mc{S}_2=(\bar{p},1]$.\footnote{It is not important whether $\bar{p}$ belongs to $\mc{S}_1$ or $\mc{S}_2$, since essentially the DM is indifferent between two arms when her belief is $\bar{p}$. However, since we laid out the HJB equation on $\mc{S}_2$, it is preferred to have an open set as the domain of the differential equation.} So far, we know that $v$ is separately convex on $\mc{S}_1$ and $\mc{S}_2$. To show that convexity is preserved on the whole region $[0,1]$, we pick the arbitrary points $p_1 \in \mc{S}_1$ and $p_2 \in \mc{S}_2$ and an arbitrary mixing weight $\xi \in (0,1)$. Define $p_\xi:= \xi p_2+(1-\xi)p_1$. If $p_\xi \in \mc{S}_1$, then $\xi v(p_2)+(1-\xi)v(p_1)$ is clearly greater than or equal to $v(p_\xi)=r$. Now suppose $p_\xi \in \mc{S}_2$, then
\begin{equation}
	\xi v(p_2)+(1-\xi)v(p_1)=\xi v(p_2)+(1-\xi)v(\bar{p})
 	\geq v\left(\xi p_2 + (1-\xi)\bar{p}\right) \geq v(p_\xi)
\end{equation}
where for the first inequality we used the convexity of $v$ on $\mc{S}_2$, and for the second one we used the monotonicity of $v$ and the fact that $p_1 \leq \bar{p}$. This concludes the global convexity of $v$, and hence the proof the theorem.\qed
\subsection{Optimal constants for value function with unambiguous information source}
\label{subsec: constants_der}
The following list is the set of all boundary conditions required for the DM's best-responding:
\begin{equation}
	\begin{split}
		&\text{(value-matching)}: r+c_1\tilde{p}^{\lambda_1}(1-\tilde{p})^{1-\lambda_1}=m(\tilde{p})-\frac{\sigma^2\delta}{2\alpha}+c_2\tilde{p}^{1-\lambda_2}(1-\tilde{p})^{\lambda_2}\\
		&\text{(smooth-pasting)}: c_1\left(\frac{\lambda_1}{\tilde{p}}-\frac{1-\lambda_1}{1-\tilde{p}}\right)\tilde{p}^{\lambda_1}(1-\tilde{p})^{1-\lambda_1}\\
		&\hspace{200pt}=\left(\overline{\theta}-\underline{\theta}\right)+c_2\left(\frac{1-\lambda_2}{\tilde{p}}-\frac{\lambda_2}{1-\tilde{p}}\right)\tilde{p}^{1-\lambda_2}(1-\tilde{p})^{\lambda_2}
	\end{split}
\end{equation}
We still need a third condition to determine all unknown variables. Note that,
\begin{align}
		\d p =\left\{ \begin{array}{lr}
		\sqrt{2}p(1-p)\left[\varphi(\sigma)\d \overline{B}+\varphi(\gamma) \d \overline{W}\right] & \text{ for } p  \geq \tilde{p}\\
		\sqrt{2}p(1-p)\varphi(\gamma)\d \overline{W} & \text{ for } p < \tilde{p}
		\end{array}\right.
\end{align}
Let us call the conjectured value function $\tilde{v}$ on $\left[0,\tilde{p}\right)$ by $\tilde{v}_-$ and on $\left(\tilde{p},1\right]$ by $\tilde{v}_+$. The threshold argument means to experiment for $p > \tilde{p}$ and to stop on $p < \tilde{p}$. Now suppose the agent instead of experimenting at $\tilde{p}$ stops for a period of $\Delta t$, in which $\Delta p \approx \sqrt{2}p(1-p)\varphi(\gamma)\sqrt{\Delta t}$, because the Bayesian formulae in the experimentation regime no longer applies. Then, the net gain from this deviation at $p =\tilde{p}$ must be negative if $\tilde{p}$ is the optimal cut-off point. The following variational analysis implies the gain for such a deviation:
\begin{equation}
	\begin{split}
		&\delta r \Delta t +\left(1-\delta \Delta t\right)\left[\frac{1}{2}v_-(p-\Delta p)+\frac{1}{2}v_+(p+\Delta p)\right]-v(p)\\
		&\approx \delta r \Delta t-v(p)\\
		& + \left(1-\delta \Delta t\right)\left\{ \frac{1}{2}\left[v_-(p)-(\Delta p) v'_-(p)+\frac{1}{2}(\Delta p)^2 v''_-(p)\right]\right.\\
		&\hspace{80pt}\left.+\frac{1}{2}\left[v_+(p)+(\Delta p) v'_+(p)+\frac{1}{2}(\Delta p)^2 v''_+(p)\right]\right\}\\
		&=\frac{1}{2}(1-\tilde{p})^2 \tilde{p}^2\varphi(\gamma)^2\left(\tilde{v}''_-(\tilde{p})+v''_+(\tilde{p})\right)\Delta t-\delta\left(\tilde{v}(\title{p})-r\right)\Delta t \leq 0,
	\end{split}
\end{equation}
therefore, 
\begin{equation}
	\begin{split}
		\frac{1}{2}(1-\tilde{p})^2 \tilde{p}^2\varphi(\gamma)^2 \left(\tilde{v}''_-(\tilde{p})+v''_+(\tilde{p})\right) &\leq \delta\left(\tilde{v}(\title{p})-r\right)=\varphi(\gamma)^2\tilde{p}^2(1-\tilde{p})^2\tilde{v}''_-(\tilde{p})\\
		 \Rightarrow \frac{1}{2}\left(\tilde{v}''_-(\tilde{p})+v''_+(\tilde{p})\right) &\leq \tilde{v}''_-(\tilde{p}).
	\end{split}
\end{equation}
By a mirror argument one can see
\begin{equation}
	\frac{1}{2}\left(\tilde{v}''_-(\tilde{p})+v''_+(\tilde{p})\right) \leq \tilde{v}''_+(\tilde{p}).
\end{equation}
Consequently it holds that $\tilde{v}''_-(\tilde{p})=\tilde{v}''_+(\tilde{p})$, leading to the super-contact condition:
\begin{equation}
	\frac{c_1\delta}{\tilde{p}^2(1-\tilde{p})^2 \varphi(\gamma)^2}\tilde{p}^{\lambda_1}(1-\tilde{p})^{1-\lambda_1}= \frac{c_2 \delta }{\tilde{p}^2(1-\tilde{p})^2 \left(\varphi(\sigma)^2+\varphi(\gamma)^2\right)}\tilde{p}^{1-\lambda_2}(1-\tilde{p})^{\lambda_2}
\end{equation}
Therefore, the value of constants $(c_1,c_2)$ are determined in terms of the cut-off point $\tilde{p}$:
\begin{equation}
	 c_1=\frac{\frac{\sigma^2}{\gamma^2} \left(r-m(\tilde{p})+\frac{\sigma^2\delta}{2\alpha}\right)}{\tilde{p}^{\lambda_1}(1-\tilde{p})^{1-\lambda_1}}, ~ c_2 = \frac{\left(1+\frac{\sigma^2}{\gamma^2}\right) \left(r-m(\tilde{p})+\frac{\sigma^2\delta}{2\alpha}\right)}{\tilde{p}^{1-\lambda_2}(1-\tilde{p})^{\lambda_2}}
\end{equation}


\subsection{Proof of proposition \ref{prop: Compar_cutoff}}
The associated equations for the cut-off probabilities in each case are expressed in \eqref{eq: pbar} and \eqref{eq: ptilde}. First, we show that $\Lambda \geq \lambda$ for any combination of variables. One can easily check from definition of $\Lambda$ and $\lambda$ that $\Lambda \geq \lambda$ iff
\begin{equation}
\label{appeq: Lgl}
	\frac{\sigma^2}{\gamma^2}\sqrt{1+\beta \gamma^2}+\left(1+\frac{\sigma^2}{\gamma^2}\right)\sqrt{1+\beta \frac{\sigma^2\gamma^2}{\sigma^2+\gamma^2}}\geq \sqrt{1+\beta \sigma^2},	
\end{equation}
in that we denote $\beta:=8\delta/ (\overline{\theta}-\underline{\theta})^2$. Because $\gamma^2 \geq \sigma^2\gamma^2/(\sigma^2+\gamma^2)$ the \textit{lhs} is larger than
\begin{equation}
	\left(1+\frac{2\sigma^2}{\gamma^2}\right)\sqrt{1+\beta \frac{\sigma^2\gamma^2}{\sigma^2+\gamma^2}}.	
\end{equation}
Therefore, a sufficient condition for \eqref{appeq: Lgl} to hold is $\left(1+\frac{2\sigma^2}{\gamma^2}\right)^2 \geq \frac{\left(1+\frac{\sigma^2}{\gamma^2}\right)\left(1+\beta \sigma^2\right)}{1+\frac{\sigma^2}{\gamma^2}+\beta \sigma^2}$, which holds because
\begin{equation}
	\left(1+\frac{2\sigma^2}{\gamma^2}\right) \geq 1\geq \frac{1+\beta \sigma^2}{1+\frac{\sigma^2}{\gamma^2}+\beta \sigma^2}.	
\end{equation}
Now we verify that $\tilde{p} \geq \bar{p}$. For this, note that $\tilde{p}=0$ if $\Lambda \leq \eta$, in that case $\lambda \leq \eta$ which implies $\bar{p}=0$. For the region $\lambda > \eta$ both cut-offs are positive. They are equal to one if $\eta \geq 1$ and are strictly smaller than one if $\eta < 1$, in that case $\tilde{p} \geq \bar{p}$ because $\Lambda \geq \lambda$.\qed

\bibliographystyle{apalike}
\bibliography{ref}

\end{document}